\documentclass[10pt,aps,pra,twocolumn,floatfix,nofootinbib,superscriptaddress,longbibliography,nobalancelastpage]{revtex4-2}
\usepackage{placeins}

\usepackage{bm,graphicx,mathrsfs,amsmath,amssymb,mathtools,makecell,bbm, amsthm,dsfont,color,nicefrac,framed,enumitem,tikz,physics,wrapfig,amsfonts,tcolorbox,times,txfonts,lipsum,nicematrix}
\newtheorem{corollary}{Corollary}
\usepackage{comment}

\renewcommand{\ket}[1]{\lvert #1\rangle}
\renewcommand{\bra}[1]{\langle #1\rvert}

\newtheorem{theorem}{Theorem}
\newtheorem{definition}{Definition}

\newcommand{\iden}{\mathbbm{1}}

\definecolor{alecolor}{RGB}{110,190,110}      
\definecolor{changescolor}{RGB}{0,120,0}

\definecolor{mycitecolor}{RGB}{204,170,0}     
\definecolor{mylinkcolor}{RGB}{110,190,110}   
\definecolor{myurlcolor}{RGB}{110,190,110}    

\usepackage[
    colorlinks=true,
    citecolor=mylinkcolor,
    linkcolor=mycitecolor,
    urlcolor=myurlcolor
]{hyperref}\usepackage{ragged2e}
\definecolor{equationcolor}{RGB}{154,135,198}
\usetikzlibrary{arrows.meta}
\usetikzlibrary{positioning}
\usetikzlibrary{calc}

\begin{document}
\title{The Prepare and Broadcast Scenario}

\author{Tailan S. Sarubi}
\affiliation{Physics Department, Federal University of Rio Grande do Norte, Natal, 59072-970, Rio Grande do Norte, Brazil}
\affiliation{International Institute of Physics, Federal University of Rio Grande do Norte, 59078-970, Natal, Brazil}
\author{Moisés Alves}
\affiliation{Physics Department, Federal University of Rio Grande do Norte, Natal, 59072-970, Rio Grande do Norte, Brazil}
\affiliation{International Institute of Physics, Federal University of Rio Grande do Norte, 59078-970, Natal, Brazil}
\affiliation{QuIIN - Quantum Industrial Innovation, EMBRAPII CIMATEC Competence Center in Quantum Technologies, SENAI CIMATEC, Av. Orlando Gomes 1845, 41650-010, Salvador, BA, Brazil. }
\author{Santiago Zamora}
\affiliation{International Institute of Physics, Federal University of Rio Grande do Norte, 59078-970, Natal, Brazil}
\author{Vinícius F. Alves}
\affiliation{Physics Department, Federal University of Rio Grande do Norte, Natal, 59072-970, Rio Grande do Norte, Brazil}
\affiliation{International Institute of Physics, Federal University of Rio Grande do Norte, 59078-970, Natal, Brazil}
\author{A.~de~Oliveira~Junior}
\affiliation{Center for Macroscopic Quantum States bigQ, Department of Physics,
Technical University of Denmark, Fysikvej 307, 2800 Kgs. Lyngby, Denmark}
\author{Rafael Chaves}
\affiliation{International Institute of Physics, Federal University of Rio Grande do Norte, 59078-970, Natal, Brazil}

\begin{abstract}
We introduce the dimension-restricted prepare and broadcast (PAB) scenario, which generalizes standard prepare-and-measure frameworks. Here, the system prepared by a sender undergoes a broadcasting transformation before being locally measured by multiple receivers. We develop a hierarchy of classical, quantum, and nonsignalling models describing this scenario, characterize their corresponding correlation sets, and derive new families of Bell-like inequalities together with linear and semidefinite programming methods for their certification. First, assuming shared randomness, we prove that the hierarchy collapses into a single set whenever we consider only one measurement per party. Then, considering multiple possible measurements, we show that PAB scenarios can activate nonclassicality of noisy broadcasting resources and can reveal preparation nonclassicality more robustly than large finite families of standard prepare-and-measure tests.
\end{abstract}

\maketitle
\section{Introduction}

Prepare-and-measure (PAM) scenarios constitute one of the most fundamental frameworks in quantum information theory. In these scenarios, a sender (Alice) encodes classical or quantum data into a physical system and sends it to a receiver (Bob), who measures it to extract information. This simple setup underlies many protocols~\cite{bohr2026quantum}, including quantum random access codes (QRACs)~\cite{Ambainis2002,Nayak1999,Ambainis2008QRAC,pawlowski2010entanglement,tavakoli2015quantum}, dense coding~\cite{moreno2021semi,tavakoli2021correlations}, dimension witnessing~\cite{gallego2010device,tavakoli2015quantum}, entanglement certification~\cite{tavakoli2018semi,moreno2021semi,pauwels2022entanglement,vieira2023interplays}, and non-stabilizerness certification~\cite{zamora2025semi}. It also forms the basis of semi-device-independent randomness generation~\cite{li2012semi,lunghi2015self,passaro2015optimal,alves2026semi} and cryptographic protocols~\cite{pawlowski2011semi,woodhead2014imperfections}. Beyond these applications, PAM scenarios play a central role in the foundations of quantum theory~\cite{pawlowski2009information,Sarubi2026,chaves2015information,chaves2018causal}, providing a simple and operationally clear setting in which quantum advantages over classical systems can be rigorously characterized.

A defining feature of standard PAM scenarios is their simple causal structure, in which a prepared system is sent directly to a measurement device, with no intermediate processing or redistribution. This simplicity makes the framework analytically tractable, but also restrictive. Since a single receiver acts on the system, multipartite correlations and distributed information processing lie outside its scope. Prepare-transform-measure extensions address this limitation by allowing intermediate parties to operate on the communicated system before forwarding it~\cite{bowles2015testing}. In these scenarios, QRACs can be implemented sequentially with multiple receivers, revealing trade-offs between their information gains, with applications to the self-testing of quantum instruments~\cite{mohan2019sequential,miklin2020semi}.
\begin{figure*} 
    \centering
    \includegraphics{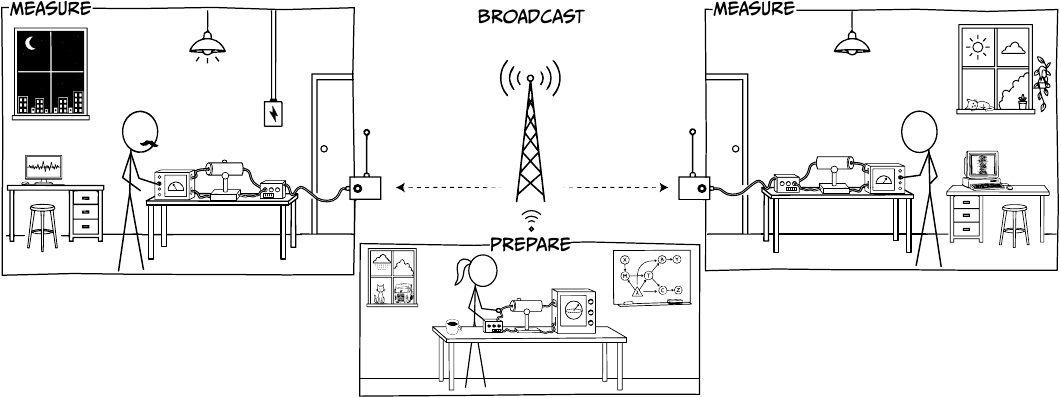}
    \caption{\textbf{Prepare and broadcast scenario}. In the center, Alice prepares a physical system according to her chosen input. The system is then sent through a broadcasting channel, which distributes its information to two distant receivers, Bob and Charlie. Each receiver performs a local measurement in their own lab and records an outcome. By repeating the experiment many times, the parties can study the correlations between Alice’s preparation and the two receivers’ results. Illustration inspired by XKCD~\cite{XKCD}.}
    \label{F-PAB-cartoon}
\end{figure*}

Relatedly, broadcasting scenarios, in which a local transformation maps a single system into a multipartite one and distributes its information among several observers, have been used to reveal nonclassical features that are hidden in standard settings. In Bell scenarios, such transformations can activate nonlocal correlations~\cite{bowles2021single,villegas2024nonlocality,boghiu2023device} and strengthen device-independent certification~\cite{polino2024experimental}. Broadcasting transformations have also been combined with prepare-and-measure scenarios in Refs.~\cite{wang2019characterising,ioannou2022receiver,jia2025characterizing} using a generalized NPA hierarchy~\cite{navascues2007bounding}. However, these approaches require prior knowledge of the inner products between the prepared quantum states, and therefore do not apply when the dimension is the only assumption.

This leaves open the dimension-bounded broadcasting problem. Once a prepared system is redistributed to several receivers, one must determine whether nonclassicality comes from the communicated system itself, from the broadcasting transformation, or from both, without relying on detailed knowledge of the preparations. This motivates the central question of this work. How can one certify nonclassicality in a dimension-bounded communication scenario once the communicated system is redistributed to several receivers, without assuming detailed knowledge of the prepared states?

We address this question by introducing the dimension-restricted prepare and broadcast (PAB) scenario (see Fig.~\ref{F-PAB-cartoon} for a pictorial representation). In this setting, Alice sends a bounded-dimensional system to a broadcasting transformation, which redistributes its information to two receivers, Bob and Charlie, who then perform local measurements and generate correlations. The key point is that the communicated system and the broadcasting resource must be treated separately. A violation of a fully classical PAB inequality certifies nonclassicality somewhere in the experiment, but not necessarily in Alice's message. To distinguish these possibilities, we introduce a hierarchy of models in which Alice's message and the broadcaster's resource can be classical, quantum, or nonsignalling. The corresponding hybrid models provide the relevant benchmarks for deciding whether an observed violation can be explained by a classical message assisted by a nonclassical broadcaster, or instead requires a genuinely quantum communicated system.

We characterize these models through PAB inequalities and optimization methods. For the fully classical and classical--nonsignalling models, we use polytope methods to derive separating inequalities, while for the classical--quantum and fully quantum models we develop semidefinite-programming relaxations to bound their achievable values. Within this framework, we prove an input-free collapse theorem: if the receivers have no measurement choices and only shared randomness is allowed, the broadcasting stage reduces to a single effective POVM on Alice's system, so quantum messages give no advantage over classical messages of the same dimension. We then show that receiver inputs restore nonclassicality and allow activation effects. In particular, PAB inequalities can reveal nonclassicality of noisy broadcasting channels below standard Bell-type thresholds, and can certify nonclassicality of Alice's communicated system beyond hybrid classical-message benchmarks.

The paper is organized as follows. In Sec.~\ref{thepam}, we review the standard PAM scenario. In Sec.~\ref{sec:pab}, we introduce the PAB scenario and its hierarchy of classical, quantum, and nonsignalling models. In Sec.~\ref{sec:inequalities}, we derive PAB inequalities and present the corresponding linear and semidefinite programming methods. In Sec.~\ref{sec:A NO-GO THEOREM}, we prove the input-free collapse theorem and discuss the entanglement-assisted case. In Sec.~\ref{sec:quantum-broadcasting},  we demonstrate activation of nonclassicality in the broadcasting channel and study the enhanced detection of nonclassicality in the prepared states. We conclude in Sec.~\ref{secFINAL REMARKS:VIOL}.

\section{A brief review of the prepare-and-measure scenario}\label{thepam}

The prepare-and-measure scenario consists of a preparation device and a measurement device. In each run, the preparer receives a classical input $x\in\mathcal X$ and, as a function of that input, generates a physical system that is sent to the measurement device. Classically, this system is represented  by a message $m\in\mathcal M$ and the measurement device decodes it based on a chosen input $y\in\mathcal Y$, producing an output $b\in\mathcal B$. The observed behavior of the experiment is therefore described by the conditional distribution $p(b\vert x,y)$ (see left panel of Fig.~\ref{fig:PAM_scenario}). The corresponding causal structure is represented by the directed acyclic graph (DAG) shown in the right panel of Fig.~\ref{fig:PAM_scenario}. As a consequence, it imposes a probability decomposition of the form


\begin{equation}
p(b\vert x,y)=\sum_{\lambda,m} p(\lambda)\,p(m\vert x,\lambda)\,p(b\vert m,y,\lambda).
\label{eq:classical_pam}
\end{equation}

As is well-known, this model is trivial  without additional restrictions since the preparer could encode the entire input $x$ into the transmitted message. To exclude this possibility, several restrictions are possible, ranging from dimensional~\cite{gallego2010device,pauwels2022almost}, energy~\cite{van2017semi}, and informational~\cite{chaves2015device,tavakoli2020informationally} restrictions on the prepared states to relational constraints limiting their fidelities~\cite{tavakoli2021semi} or overlaps~\cite{wang2019characterising}. Here, we assume a constraint on the dimension $d$ of the communicated system, requiring the message alphabet to satisfy $\vert\mathcal M\vert = d<|\mathcal X|$. For fixed input and output alphabets, the set of classical distributions compatible with this restriction forms a convex polytope.

When generalizing to a quantum scenario, we restrict our attention to the case of quantum messages  while maintaining the classical shared resource and other causal constraints intact. More precisely, we are interested in the case where for each input $x$, the preparation device emits a quantum state $\rho_x\in\mathcal D(\mathbb C^d)$, while for each measurement choice $y$, the measurement device implements a POVM $\{M_{b\vert y}\}_{b\in\mathcal B}$. The resulting conditional probabilities are given by the Born rule 
\begin{equation}
    p(b\vert x,y)=\tr\big(\rho_x M_{b\vert y}\big).
\end{equation}
In what follows, we generalize this scenario to a multi-receiver case.

\begin{figure}[t!]
\centering
\includegraphics{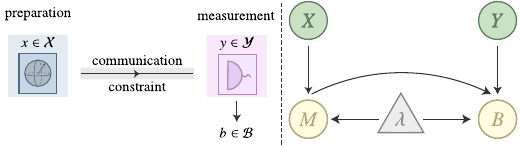}
\caption{\textbf{Prepare-and-measure scenario.}
Left: in each run, the preparation device receives an input $x\in\mathcal X$ and sends a physical system to the measurement device, subject to a communication constraint. The measurement device receives an input $y\in\mathcal Y$ and produces an output $b\in\mathcal B$. Right: DAG representation of the corresponding causal structure. The influence of the preparation choice $X$ on the outcome $B$ is mediated by the system sent by $M$, while the measurement choice $Y$ acts locally on the measurement device. The latent variable $\lambda$ represents shared randomness between the devices and is assumed to be independent of the freely chosen inputs $X$ and $Y$. The observed behavior is described by $p(b\vert x,y)$.}
\label{fig:PAM_scenario}
\end{figure}

\section{The prepare and broadcast scenario}
\label{sec:pab}

The prepare and broadcast (PAB) scenario extends the standard prepare-and-measure setting by replacing the single receiver with several receivers. While an $n$-receiver generalization is theoretically natural, the exact characterization of the corresponding correlation sets quickly becomes computationally intractable as the number of parties increases. Consequently, in this work, we focus on the simplest nontrivial case of two receivers, Bob and Charlie. Alice prepares a system that is sent to a broadcasting device. The broadcasting device then distributes the information to Bob and Charlie, who choose their inputs locally, $(y,z)$, and produce outputs $(b,c)$.

A simple way to picture the broadcasting channel is as a physical device with one input port and two output ports. It receives the system prepared by Alice and converts it into two systems, one delivered to Bob and one to Charlie. In a classical implementation, this could be a router, transmitter, or copying device that sends the same message, or noisy versions of it, to two different users. In a quantum implementation, however, the device should not be understood as a perfect copier of an arbitrary unknown state. Rather, it may be an interaction with an ancillary system, a beam splitter, an approximate cloning transformation, or more generally any physical process that redistributes the information contained in Alice’s input into the two output systems. Bob and Charlie then access only their local shares, and the experiment probes how Alice’s preparation is encoded jointly in their individual outcomes and in the correlations between them.

We use the labels $\mathcal{C}$, $\mathcal{Q}$, and $\mathcal{NS}$ to specify whether a resource is classical, quantum, or nonsignalling. In the PAB scenario, two labels are needed. The first one refers to the type of system sent by Alice to the broadcaster, while the second one refers to the type of system output by the broadcaster to generate correlations between Bob and Charlie. Thus, $\mathcal{CC}$ denotes a classical message followed by classical broadcasting, $\mathcal{CQ}$ denotes a classical message followed by a quantum broadcaster, $\mathcal{QC}$ denotes a quantum message followed by a classical broadcaster, and $\mathcal{CNS}$ denotes a classical message followed by a nonsignalling broadcaster. The fully quantum model is denoted by $\mathcal{QQ}$.

In the fully classical model, in each run Alice receives an input $x\in\{0,\ldots,n_A-1\}$ and sends a classical message $m\in\{0,\ldots,d-1\}$. Bob chooses an input $y\in\{0,\ldots,n_B-1\}$ and outputs $b\in\{0,\ldots,o_B-1\}$, while Charlie chooses an input $z\in\{0,\ldots,n_C-1\}$ and outputs $c\in\{0,\ldots,o_C-1\}$. In general, we allow for shared randomness $\lambda$ between the preparation, the broadcaster and the two measurement devices.  The causal structure is shown in Fig.~\ref{fig:broadcast}. A PAB behavior $p(b,c\vert x,y,z)$ is compatible with the fully classical broadcasting model $\mathcal{CC}$ if it can be written as
\begin{align} 
p_{\mathcal{CC}}(b,c\vert x,y,z)=\sum_{m,t,\lambda}p(\lambda)&p(m\vert x,\lambda)p(t\vert m,\lambda) \: \times \nonumber \\ &\times p(b\vert t,y,\lambda)p(c\vert t,z,\lambda).
\label{eq:broadcast-LHV} 
\end{align} 
Here $p(\lambda)$ is the hidden-variable distribution, $p(m\vert x,\lambda)$ is Alice's preparation map, $p(t\vert m,\lambda)$ is the classical broadcasting map, and $p(b\vert t,y,\lambda)$ and $p(c\vert t,z,\lambda)$ are Bob's and Charlie's local response functions. The message $m$ has bounded cardinality  $\vert\mathcal M\vert = d<\vert\mathcal X\vert$. In contrast, in the model above we do not impose a bound on the cardinality of the broadcasted variable $t$. If one wants to study restricted broadcasting, then one should instead introduce separate broadcasted variables $t_B$ and $t_C$, with their own cardinality constraints, for Bob and Charlie.

\begin{figure}[t!]
\centering
\includegraphics{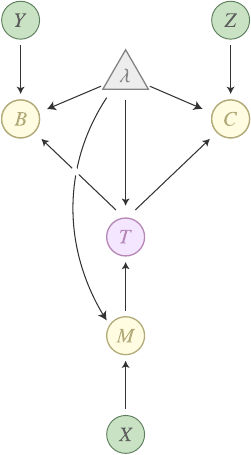}
\caption{\textbf{Prepare and broadcast DAG.} Alice receives an input $x$ and prepares a bounded classical message $m$. The message is processed by the broadcasting map, which produces the classical variable $t$ distributed to Bob and Charlie. Bob receives input $y$ and outputs $b$, while Charlie receives input $z$ and outputs $c$. The hidden variable $\lambda$ may influence Alice's preparation, the broadcasting map, and the local response functions of Bob and Charlie. The observed behavior is the conditional distribution $p(b,c\vert x,y,z)$.}
\label{fig:broadcast}
\end{figure}

For the unrestricted case considered here, the variable $t$ can be removed without loss of generality. Specifically, after absorbing all local randomness into $\lambda$, the maps can be taken to be deterministic. The broadcasting map is then a deterministic function $t=g_\lambda(m)$. Since Bob and Charlie also have access to $\lambda$, they can compute $g_\lambda(m)$ locally. Thus the function $g_\lambda$ can be absorbed into their response functions. We can thus use the simpler deterministic form 
\begin{equation} 
\!\!p_{\mathcal{CC}}(b,c\vert x,y,z)=\sum_{m,\lambda}p(\lambda)D^A_{\lambda}(m\vert x)D^B_{\lambda}(b\vert y,m)D^C_{\lambda}(c\vert z,m), \label{eq:broadcast-LHV-Det-simplified} 
\end{equation} 
where $D^A_\lambda$, $D^B_\lambda$, and $D^C_\lambda$ are deterministic response functions. This form is equivalent to Alice sending the dimension-bounded classical message $m$ directly to both Bob and Charlie.

For fixed finite alphabets, Eq.~\eqref{eq:broadcast-LHV-Det-simplified} defines a polytope. Its facets give PAB inequalities of the form \begin{equation}
\sum_{b,c,x,y,z}\alpha_{b,c,x,y,z}\,p(b,c\vert x,y,z)\le\beta_{\mathcal{CC}}, 
\label{eq:classicalinequalities} 
\end{equation}
where $\alpha_{b,c,x,y,z}$ are real coefficients and $\beta_{\mathcal{CC}}$ is the maximum value over the fully classical model. In Sec.~\ref{sec:inequalities}, we obtain such inequalities using standard facet-enumeration methods. The PAB scenario contains both PAM and Bell scenarios as special cases. If Alice has only one input, or if $x$ is fixed, the remaining behavior of Bob and Charlie is a standard bipartite Bell scenario. On the other hand, by marginalizing over Charlie's output and ignoring the corresponding input, one obtains
\begin{equation} 
p(b\vert x,y)=\sum_c p(b,c\vert x,y,z), 
\end{equation} 
which is independent of $z$ in the models considered here. This gives a standard PAM scenario between Alice and Bob. 

This observation is useful, but it also creates a possible ambiguity. Some valid inequalities for the classical PAB model are simply lifted Bell inequalities for Bob and Charlie. A violation of such an inequality shows that the broadcasting part is nonclassical, but it does not by itself show that Alice's preparation is nonclassical. For example, a quantum broadcaster could ignore Alice's input and distribute an entangled state to Bob and Charlie. As a result, if the goal is to certify nonclassicality of Alice's preparation or message, the relevant witness should remain valid even when the broadcaster is allowed to use nonclassical resources. This motivates hybrid models.

First, consider the case in which Alice's message is still classical, but the broadcaster is allowed to distribute a quantum state to Bob and Charlie. The resulting $\mathcal{CQ}$ model is
\begin{align} 
p_{\mathcal{CQ}}(b,c\vert x,y,z)=\sum_{m,\lambda}p(\lambda)p(m\vert x,\lambda)p^{Q}(b,c\vert y,z,m,\lambda), 
\label{eq:broadcast-LHV-Q} 
\end{align} 
where, for each fixed $(m,\lambda)$, 
\begin{equation}
p^{Q}(b,c\vert y,z,m,\lambda)=\tr\!\left[\rho^{BC}_{m,\lambda}\left(M^B_{b\vert y,\lambda}\otimes M^C_{c\vert z,\lambda}\right)\right]. \label{eq:CQ_block} \end{equation} 
Here $\rho^{BC}_{m,\lambda}$ is a bipartite quantum state prepared by the broadcaster, depending on the classical message $m$ and the hidden variable $\lambda$. For a given inequality with coefficients $\alpha_{b,c,x,y,z}$, the relevant bound for this model is $\beta_{\mathcal{CQ}}$, the maximum over all behaviors of the form given in Eq.~\eqref{eq:broadcast-LHV-Q}. A violation above $\beta_{\mathcal{CQ}}$ certifies that the message sent by Alice cannot be classical, even if the broadcaster is allowed to distribute quantum correlations. As discussed in Sec.~\ref{subsec:cqineqs}, this bound can be upper-bounded using an adapted Navascués--Pironio--Acín hierarchy~\cite{navascues2007bounding}.

The fully quantum model $\mathcal{QQ}$ replaces Alice's classical message by a quantum state. In its most general form, allowing shared randomness, the correlations can be written as 
\begin{equation}
\!\!p_{\mathcal{QQ}}(b,c\vert x,y,z)=\sum_{\lambda}p(\lambda)\tr\!\left[\mathcal T_\lambda(\rho_{x,\lambda})\left(M^B_{b\vert y,\lambda}\otimes M^C_{c\vert z,\lambda}\right)\right], \label{eq:born_broadcasting} 
\end{equation} 
where $\rho_{x,\lambda}$ is a state on a $d$-dimensional Hilbert space $\mathcal H_M$, and $\mathcal T_\lambda:\mathcal L(\mathcal H_M)\to \mathcal L(\mathcal H_B\otimes\mathcal H_C)$ is a quantum broadcasting channel. The output spaces $\mathcal H_B$ and $\mathcal H_C$ are not assumed to have bounded dimension. By Stinespring dilation, each channel $\mathcal T_\lambda$ can be represented by an isometry into a larger Hilbert space, followed by tracing out an environment. In Sec.~\ref{subsec:qqineqs}, we use a hybrid SDP hierarchy combining NPA-type moments~\cite{navascues2007bounding,navascues2008convergent} with finite-dimensional techniques~\cite{navascues2015bounding,navascues2015characterizing,tavakoli2021correlations} to upper-bound the maximal violation achievable in this model.

One can also consider a stronger relaxation of the classical-message model, in which the broadcaster is allowed to distribute any nonsignalling resource to Bob and Charlie. This gives the $\mathcal{CNS}$ model, 
\begin{align} 
\!\!p_{\mathcal{CNS}}(b,c\vert x,y,z)=\sum_{m,\lambda}p(\lambda)p(m\vert x,\lambda)p^{\mathcal{NS}}_{BC}(b,c\vert y,z,m,\lambda). \label{eq:broadcast-LHV-NS} \end{align} 
For each fixed $(m,\lambda)$, the distribution $p^{\mathcal{NS}}_{BC}$ must be nonsignalling. That is, 
\begin{align}
\sum_b p^{\mathcal{NS}}_{BC}(b,c\vert y,z,m,\lambda)&=\sum_b p^{\mathcal{NS}}_{BC}(b,c\vert y',z,m,\lambda), \label{eq:NS1}\\ \sum_c p^{\mathcal{NS}}_{BC}(b,c\vert y,z,m,\lambda)&=\sum_c p^{\mathcal{NS}}_{BC}(b,c\vert y,z',m,\lambda), \label{eq:NS2}
\end{align} 
for all allowed values of the indices. In addition, 
\begin{align} p^{\mathcal{NS}}_{BC}(b,c\vert y,z,m,\lambda)&\ge 0, \label{eq:NSpos}\\ \sum_{b,c}p^{\mathcal{NS}}_{BC}(b,c\vert y,z,m,\lambda)&=1. \label{eq:NSnorm} 
\end{align} 
This model allows Bob and Charlie to share arbitrary nonsignalling correlations conditioned on the classical message $m$, but it still forbids direct communication between them. The set defined by Eq.~\eqref{eq:broadcast-LHV-NS} has the advantage of being a polytope for a finite alphabet. In Sec.~\ref{subsec:cnsineqs}, we use this polytope to derive inequalities that are valid even under this nonsignalling relaxation.

Finally, one may consider the $\mathcal{QC}$ model, where Alice sends a quantum system but the broadcaster is restricted to classical random variables. In this case, the $\mathcal{QC}$ and $\mathcal{CC}$ sets coincide. To see this, suppose the broadcaster measures Alice's $d$-dimensional quantum system and obtains a classical outcome $t$. The remaining part of the protocol is just classical post-processing of $t$ into Bob's and Charlie's outputs. Thus the only nontrivial object is the conditional distribution $p(t\vert x)$ generated by a single measurement on $d$-dimensional quantum states. By the Frenkel--Weiner theorem~\cite{frenkel2015classical}, any such single-measurement prepare-and-measure behavior can be simulated by a classical message with $d$ possible values, assisted by shared randomness. Hence, every $\mathcal{QC}$ behavior admits a $\mathcal{CC}$ model. The converse inclusion is immediate, since a classical $d$-valued message can be encoded into $d$ orthogonal quantum states and then measured by the classical broadcaster. Therefore, whenever the broadcasting device is classical, allowing Alice to send quantum states gives no advantage over classical communication alone.

We do not introduce a $\mathcal{QNS}$ model because nonsignalling constraints alone do not specify how an incoming quantum state maps to a (possibly post-quantum) nonsignalling resource. Unlike the $\mathcal{CNS}$ case—where classical messages simply index nonsignalling behaviors—defining $\mathcal{QNS}$ requires an additional operational framework for quantum-to-nonsignalling transformations, which lies beyond the scope of this work.

\section{Characterizing the set of correlations in the prepare and broadcast scenario}
\label{sec:inequalities}
We now describe how we characterize several sets of correlations in the PAB scenario. There is a natural hierarchy between some of these sets. Since classical correlations are a subset of quantum correlations, and quantum correlations are a subset of nonsignalling correlations, one has
\begin{equation}
\mathcal{CC}=\mathcal{QC} \subsetneq \mathcal{CQ} \subsetneq \mathcal{CNS}.
\end{equation}
The inclusions are strict. This can already be seen by fixing Alice's input, so that the remaining behavior is an ordinary Bell scenario between Bob and Charlie. In that case, the CHSH bounds satisfy
\begin{align}
\beta_{\mathcal C}^{\mathrm{CHSH}}&=2,\\
\beta_{\mathcal Q}^{\mathrm{CHSH}}&=2\sqrt{2},\\
\beta_{\mathcal{NS}}^{\mathrm{CHSH}}&=4.
\end{align}
Thus the three possible resources for the $BC$ block lead to different sets of correlations. There is also the chain
\begin{equation}
\mathcal{CC}=\mathcal{QC} \subsetneq \mathcal{CQ} \subsetneq \mathcal{QQ}.
\end{equation}
The inclusion $\mathcal{CQ}\subsetneq\mathcal{QQ}$ follows because a classical message can be encoded into orthogonal quantum states, and the quantum broadcasting channel can then prepare the corresponding bipartite state for Bob and Charlie. The inclusion is strict because, after marginalizing over one receiver, the $\mathcal{QQ}$ model reduces to an ordinary quantum prepare-and-measure scenario, while the $\mathcal{CQ}$ model reduces to a classical $d$-message PAM model. Quantum dimension witnesses therefore separate the two sets.

The sets $\mathcal{QQ}$ and $\mathcal{CNS}$ are not ordered by inclusion. For example, $\mathcal{CNS}$ contains post-quantum nonsignalling correlations for the $BC$ block, such as PR-box correlations, which cannot be reproduced by $\mathcal{QQ}$. Conversely, as will be shown in Sec.~\ref{subsec:qqineqs}, $\mathcal{QQ}$ contains PAB behaviors whose marginal over one receiver yields a quantum PAM correlation that cannot be reproduced by a classical $d$-message model. In contrast, the same marginalization of any $\mathcal{CNS}$ behavior always gives a classical $d$-message PAM behavior. A schematic representation of these relations is shown in Fig.~\ref{fig:sets}. 

\begin{figure}[t]
\centering
\includegraphics{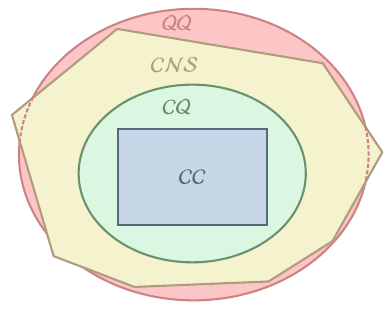}
\caption{\textbf{Correlation sets in the prepare and broadcast scenario.} The fully classical set $\mathcal{CC}$ is contained in the classical--quantum set $\mathcal{CQ}$. The set $\mathcal{CQ}$ is contained in both the nonsignalling relaxation $\mathcal{CNS}$ and the fully quantum set $\mathcal{QQ}$. The sets $\mathcal{CNS}$ and $\mathcal{QQ}$ are incomparable. The figure is schematic and not meant to represent the geometry of the sets.}
\label{fig:sets}
\end{figure}

In the following we focus on binary outputs. All models considered here are nonsignalling between Bob and Charlie for each fixed preparation input $x$. It is therefore convenient to describe the behavior $p(b,c\vert x,y,z)$ in terms of expectation values. We define
\begin{align}
\label{eq20}
\langle B_y\rangle^{[x]}
&:=
\sum_{b,c\in\{0,1\}}(-1)^b\,p(b,c\vert x,y,z),
\qquad \text{for any }z,
\nonumber\\
\langle C_z\rangle^{[x]}
&:=
\sum_{b,c\in\{0,1\}}(-1)^c\,p(b,c\vert x,y,z),
\qquad \text{for any }y,
\nonumber\\
\langle B_y C_z\rangle^{[x]}
&:=
\sum_{b,c\in\{0,1\}}(-1)^{b+c}\,
p(b,c\vert x,y,z).
\end{align}
In this representation, normalization and nonsignalling conditions are automatically satisfied, and the full conditional distribution can be reconstructed as
\begin{align}
p(b,c\vert x,y,z) &= \frac14\Big[ 1 +(-1)^b\langle B_y\rangle^{[x]} +(-1)^c\langle C_z\rangle^{[x]} \nonumber\\ &\qquad \hspace{0.3cm}+(-1)^{b+c}\langle B_yC_z\rangle^{[x]} \Big]. 
\label{eq:NS-correlator}
\end{align}
This parametrization automatically enforces normalization and nonsignalling. The remaining nontrivial condition is positivity of all probabilities in Eq.~\eqref{eq:NS-correlator}.

As mentioned before,  the fully classical set $\mathcal{CC}$ and the classical-nonsignalling set $\mathcal{CNS}$ are polytopes for finite alphabets and thus we can characterize them through their $H$-representation. We do this by standard tools of convex geometry. First, the $V$-representation is obtained by enumerating extremal strategies. In the fully classical case these are deterministic preparation and response functions. In the classical-nonsignalling case they are deterministic encodings on Alice's side, combined with extremal nonsignalling blocks for Bob and Charlie. After removing duplicate points, we use \texttt{lrs}~\cite{lrs} to compute the corresponding $H$-representations. The resulting facets define prepare-and-broadcast inequalities.

\subsection{PAB inequalities for the fully classical broadcasting model}
\label{subsec:ccineqs}

As discussed above, because the broadcasted variable $t$ is unbounded and unobserved, the deterministic broadcasting map can be absorbed directly into Bob's and Charlie's response functions. Each classical extremal strategy is then completely defined by a set of deterministic maps $m = f_\lambda(x)$, $b = h^B_\lambda(y,m)$ and $c = h^C_\lambda(z,m)$. The associated extremal behavior is simply:
 \begin{equation}
p_\lambda(b,c\vert x,y,z) = \delta_{b,\,h^B_\lambda(y,f_\lambda(x))}\,\delta_{c,\,h^C_\lambda(z,f_\lambda(x))}.
\label{eq:classical_vertex}
\end{equation} 
The fully classical set, in its V-representation, is defined as the convex hull of the vertices $\vec p_\lambda$, defined by the entries $(\vec {p}_\lambda)_{bcxyz} = p_\lambda(b,c\vert x,y,z)$. 

We can equivalently define  the $\mathcal {CC}$ polytope in the space of correlators. For binary outputs, following Eq.~\eqref{eq20}, the single-operator deterministic expectation values are,
\begin{align}
\langle B_y\rangle^{[x],\lambda} &= (-1)^{h^B_\lambda(y,f_\lambda(x))},\\
\langle C_z\rangle^{[x],\lambda} &= (-1)^{h^C_\lambda(z,f_\lambda(x))},
\end{align}
while the two-body deterministic correlators are given by
\begin{equation}
\langle B_y C_z\rangle^{[x],\lambda}=\langle B_y\rangle^{[x],\lambda}\langle C_z\rangle^{[x],\lambda}=(-1)^{h^B_\lambda(y,f_\lambda(x))+h^C_\lambda(z,f_\lambda(x))}.
\label{eq:product_structure}
\end{equation}
Therefore, the corresponding V-representation is given by the convex hull of the vertices $\vec c_\lambda$, defined by the entries $(\vec {c}_\lambda)_{ijxyz} = \langle B^i_y C^j_z\rangle^{[x],\lambda}$, where $i,j=0,1$ and $B^0=C^0 = \iden $, $B^1=B, C^1=C$.

We first consider the simplest scenario with three binary preparations and two dichotomic measurements. The deterministic encodings $f_\lambda:\mathcal X\to\{0,1\}$ give $2^3=8$ choices. For each receiver, there are $2^4=16$ deterministic response functions, since the response depends on one binary setting and one binary message. This gives $2^3\times 2^4\times 2^4=2048$ deterministic strategies, each defining an extremal behaviour of the form given in Eq.~\eqref{eq:classical_vertex}. Some of them may be repeated. After removing duplicates, the full expectation-value representation, including both one-body and two-body correlators, gives a polytope in $\mathbb R^{24}$ with $736$ distinct vertices. The full facet enumeration is already computationally demanding in this minimal PAB scenario. A partial list of full-space inequalities is available in Ref.~\cite{sarubi2025pab}. We therefore focus on two specific projections of the polytope that are particularly useful and for which full facet descriptions can be obtained.

The first projection retains only the two-body expectation values. For each preparation $x\in\{0,1,2\}$, we keep the four correlators $\langle B_y C_z\rangle^{[x]}$ with $y,z\in\{0,1\}$. This gives a projected polytope in $\mathbb R^{12}$ with $N_{\rm 2body}=176$ distinct vertices. Converting this $V$-representation into an $H$-representation gives $8400$ facet inequalities. We evaluated all 42 inequivalent two-body facet classes under the $\mathcal{CQ}$ and $\mathcal{QQ}$ models and compared their robustness to the noise models considered in Sec.~\ref{sec:quantum-broadcasting}. The inequalities highlighted below are therefore not arbitrary representatives. Among these classes, $W_{\mathcal{CC}}^{(2)}$ gives the lowest critical visibility for noise acting on the broadcasting channel, based on the $\mathcal{CC}$--$\mathcal{QQ}$ comparison, while $W_{\mathcal{CC}}^{(1)}$ gives the lowest critical visibility for noise acting on the preparations, for which the relevant comparison is $\mathcal{CQ}$--$\mathcal{QQ}$.
\begin{align}
W_{\mathcal{CC}}^{(1)} ={}&3\langle B_0C_0\rangle^{[0]}+\langle B_0C_1\rangle^{[0]}+\langle B_1C_0\rangle^{[0]}
-\langle B_1C_1\rangle^{[0]}\nonumber\\
&-2\langle B_0C_1\rangle^{[1]}-2\langle B_1C_0\rangle^{[1]} -3\langle B_0C_0\rangle^{[2]}
+\langle B_0C_1\rangle^{[2]}
\nonumber\\
&+\langle B_1C_0\rangle^{[2]}+\langle B_1C_1\rangle^{[2]}\le 8,
\label{eq:CC}
\end{align}
and
\begin{align}
W_{\mathcal{CC}}^{(2)} = {}& 2\langle B_0C_0\rangle^{[0]}+2\langle B_1C_1\rangle^{[0]} +2\langle B_0C_1\rangle^{[1]}
\nonumber\\
& -2\langle B_1C_0\rangle^{[1]}-\langle B_0C_0\rangle^{[2]} -\langle B_0C_1\rangle^{[2]}
\nonumber\\
&+\langle B_1C_0\rangle^{[2]} -\langle B_1C_1\rangle^{[2]} \le 6.
\label{eq:CC2}
\end{align}
The second projection retains only the single-body correlators. For each preparation $x\in\{0,1,2\}$, we keep the $2$ expectation values of Bob $\langle B_y\rangle^{[x]}$ and the $2$ of Charlie $\langle C_z\rangle^{[x]}$, with  $y,z\in\{0,1\}$. This gives a projected polytope in $\mathbb R^{12}$. The projected vertex set contains $N_{\rm marg}=736$ distinct vertices, and its $H$-representation contains $552$ facets. Under the full relabeling group, these facets reduce to $5$ equivalence classes, $4$ of which are nontrivial. Representatives of the nontrivial classes are
\begin{align}
W_{\mathrm{marg}}^{(1)}
={}&-\langle B_1\rangle^{[0]}+\langle C_0\rangle^{[0]}+\langle C_1\rangle^{[0]}
\nonumber\\
&+\langle B_1\rangle^{[1]}-\langle C_0\rangle^{[1]}+\langle C_1\rangle^{[1]}
\nonumber\\
&+\langle B_1\rangle^{[2]}+\langle C_0\rangle^{[2]}-\langle C_1\rangle^{[2]}\le 5,
\label{eq:Wmarg1}
\\[0.6em]
W_{\mathrm{marg}}^{(2)}
={}&\langle B_0\rangle^{[0]}-\langle B_1\rangle^{[0]}-\langle B_0\rangle^{[1]}
\nonumber\\
&+\langle C_0\rangle^{[1]}+\langle B_1\rangle^{[2]}-\langle C_0\rangle^{[2]}\le 4,
\label{eq:Wmarg2}
\\[0.6em]
W_{\mathrm{marg}}^{(3)}
={}&\langle B_1\rangle^{[0]}+\langle C_1\rangle^{[0]}-\langle B_1\rangle^{[1]}
\nonumber\\
&+\langle C_1\rangle^{[1]}-\langle C_1\rangle^{[2]}\le 3,
\label{eq:Wmarg3}
\\[0.6em]
W_{\mathrm{marg}}^{(4)}
={}&\langle C_0\rangle^{[0]}+\langle C_1\rangle^{[0]}-\langle C_0\rangle^{[1]}
\nonumber\\
&+\langle C_1\rangle^{[1]}-\langle C_1\rangle^{[2]}\le 3.
\label{eq:Wmarg4}
\end{align}
The full list of marginal representatives is given in Ref.~\cite{sarubi2025pab}. Some of these marginal inequalities are lifts of ordinary PAM witnesses. For instance, up to relabelings and output flips, $W_{\mathrm{marg}}^{(4)}$ is the usual three-preparation PAM witness written for Charlie's marginal behavior~\cite{gallego2010device}.

\subsection{PAB inequalities for the hybrid classical-NS broadcasting model}
\label{subsec:cnsineqs}

We now consider the model in which Alice sends a classical message, while the broadcaster is allowed to distribute nonsignalling correlations between Bob and Charlie. Since the preparation and message alphabets are finite, and since the nonsignalling $BC$ block defines a polytope, the classical--$\mathcal{NS}$ model also defines a polytope. Its generators are obtained by combining deterministic encodings on Alice's side with extremal nonsignalling blocks on the broadcasting side. We write these blocks in the correlator coordinates introduced in Eq.~\eqref{eq20}, which are the coordinates used in the facet enumeration.

For binary outputs, each nonsignalling block $p^{\mathcal{NS}}_{BC}$ can be parametrized in terms of correlators. For each fixed $(m,\lambda)$, we define
\begin{align}
\langle B_y\rangle^{(m,\lambda)}
&:=
\sum_{b,c\in\{0,1\}}(-1)^b\,
p^{\mathcal{NS}}_{BC}(b,c\vert y,z,m,\lambda),
\qquad \text{for any }z,
\nonumber\\
\langle C_z\rangle^{(m,\lambda)}
&:=
\sum_{b,c\in\{0,1\}}(-1)^c\,
p^{\mathcal{NS}}_{BC}(b,c\vert y,z,m,\lambda),
\qquad \text{for any }y,
\nonumber\\
\langle B_y C_z\rangle^{(m,\lambda)}
&:=
\sum_{b,c\in\{0,1\}}(-1)^{b+c}\,
p^{\mathcal{NS}}_{BC}(b,c\vert y,z,m,\lambda),
\label{eq:CNS_block_correlators}
\end{align} 
with normalization and nonsignalling automatically satisfied in this representation. However, we must still impose positivity of the probabilities
\begin{align}
p^{\mathcal{NS}}_{BC}(b,c\vert y,z,m,\lambda)
&=
\frac14\Big[
1
+(-1)^b\langle B_y\rangle^{(m,\lambda)} +(-1)^c\langle C_z\rangle^{(m,\lambda)}
\nonumber\\
&\qquad
+(-1)^{b+c}
\langle B_yC_z\rangle^{(m,\lambda)}
\Big],
\label{eq:NS-correlator-2}
\end{align}
so that the only remaining nontrivial constraint is
\begin{align}
1 +s_b\langle B_y\rangle^{(m,\lambda)} +s_c\langle C_z\rangle^{(m,\lambda)}+s_bs_c \langle B_yC_z\rangle^{(m,\lambda)} \ge 0,
\label{eq:NS-positivity}
\end{align}
for all $s_b,s_c\in\{\pm1\}$ and for all allowed values of $y,z,m,\lambda$. In the case of binary variables, each input pair $(y,z)$ gives four such inequalities, for a total of $16$ constraints per message value. Therefore, for  $m\in\{0,1\}$, this gives $32$ inequalities altogether. This provides an $H$-representation of the conditional nonsignalling set, from which the extremal vertices can be obtained by standard vertex-enumeration methods.

To construct the full broadcasting polytope, we combine the conditional nonsignalling blocks with deterministic preparation strategies. At an extremal preparation strategy, Alice's map is $p(m\vert x,\lambda)=\delta_{m,f_\lambda(x)}$, where $f_\lambda:\mathcal X\to\mathcal M$ is a deterministic encoding. For each fixed $\lambda$, the input $x$ selects the message-labeled nonsignalling block with $m=f_\lambda(x)$. The resulting behavior is $p_{\lambda}(b,c\vert x,y,z) =p^{\mathcal{NS}}_{BC} \bigl(b,c\vert y,z,m=f_\lambda(x),\lambda\bigr)$. Equivalently, in correlator form,
\begin{align}
\langle B_y\rangle^{[x],\lambda}
&=
\langle B_y\rangle^{(f_\lambda(x),\lambda)},
\nonumber\\
\langle C_z\rangle^{[x],\lambda}
&=
\langle C_z\rangle^{(f_\lambda(x),\lambda)},
\nonumber\\
\langle B_yC_z\rangle^{[x],\lambda}
&=
\langle B_yC_z\rangle^{(f_\lambda(x),\lambda)}.
\end{align}
Hence, each generator of the classical--$\mathcal{NS}$ broadcasting polytope is specified by a deterministic encoding together with one extremal nonsignalling block for each message value.

For $\vert\mathcal X\vert=3$ and $\vert\mathcal M\vert=2$, there are $2^{\vert\mathcal X\vert}=8$ deterministic encodings. Each conditional nonsignalling block has $24$ extremal vertices, namely $16$ local vertices and $8$ nonlocal vertices. Combining one block for each message value therefore gives $8\times 24^2=4608$ candidate generators. After removing duplicates, the full correlator representation contains $1680$ distinct generating points in $\mathbb R^{24}$. As in the fully classical case, a complete full-space facet enumeration is computationally demanding. A partial list is available in Ref.~\cite{sarubi2025pab}. We therefore focus on projections of the polytope onto subspaces of interest.

Projecting this set onto the two-body correlator space gives a polytope in $\mathbb R^{12}$ with $736$ distinct vertices. Its $H$-representation contains $552$ facets, grouped into $5$ equivalence classes under the full relabeling group. A representative inequality is
\begin{align}
W_{\mathcal{CNS}}^{(3)}
={}&
\langle B_0 C_0\rangle^{[0]}
-\langle B_0 C_1\rangle^{[0]}
+\langle B_0 C_1\rangle^{[1]}
\nonumber\\
&-\langle B_1 C_1\rangle^{[1]}
-\langle B_0 C_0\rangle^{[2]}
+\langle B_1 C_1\rangle^{[2]}
\le 4 .
\label{eq:WNS3}
\end{align}
The full list of representatives is given in Ref.~\cite{sarubi2025pab}.

 Moving on to the marginal polytope,  projecting the $4608$ candidate classical--$\mathcal{NS}$ generators onto the $12$-dimensional marginal space gives $833$ distinct projected points before removing redundancies. The resulting $H$-representation contains the same marginal facets as in the fully classical case. In particular, up to relabelings, the nontrivial classes are represented by Eqs.~\eqref{eq:Wmarg1}--\eqref{eq:Wmarg4}. This is expected, since allowing the $BC$ block to be nonsignalling changes the possible joint correlations between Bob and Charlie, but it does not enlarge the set of local marginal behaviors generated from a bounded classical message.

\subsection{PAB inequalities for the classical-quantum broadcasting model}
\label{subsec:cqineqs}

We now consider bounds for the classical--quantum model in Eq.~\eqref{eq:broadcast-LHV-Q}. In this model Alice still sends a classical message, but the broadcaster is allowed to distribute a quantum state to Bob and Charlie. Given a linear PAB witness, the goal is to compute, or upper bound, its maximal value over this model. Since the witness is linear and shared randomness only convexifies the set, it is enough to optimize over deterministic encodings $f:\mathcal X\to\mathcal M$. For a fixed encoding $f$, all inputs $x$ mapped to the same message $m=f(x)$ must give the same bipartite quantum behavior. For example, if $\vert\mathcal X\vert=3$, $\vert\mathcal M\vert=2$, and $f(0)=0$, $f(1)=f(2)=1$, then the broadcaster cannot distinguish between the inputs $x=1$ and $x=2$. Thus, $p(b,c\vert x=1,y,z)=p(b,c\vert x=2,y,z)$ for all $b,c,y,z$ which  implies that all corresponding expectation values must also be equal, $\langle B_yC_z\rangle^{[1]}=\langle B_yC_z\rangle^{[2]}$  for all  $y,z$.

We focus on the two-body witnesses obtained in this work, written as 
\begin{equation}
W=\sum_{x,y,z} \alpha_{xyz} \langle B_yC_z\rangle^{[x]} .\label{eq:General_two_Corr_ineq}  
\end{equation}
For a fixed deterministic encoding $f$, the coefficients associated with the same message can be grouped as $A^{(m\vert f)}_{yz} = \sum_{x:\,f(x)=m} \alpha_{xyz}$. The contribution of the message $m$ is then the Bell-type expression $\sum_{y,z} A^{(m\vert f)}_{yz} \langle B_yC_z\rangle_m ,$ where $\langle B_yC_z\rangle_m$ denotes the correlator of the bipartite quantum state distributed by the broadcaster when the classical message is $m$.

For fixed coefficients $A=\{A_{yz}\}$, define the corresponding bipartite quantum value as
\begin{equation}
\beta_Q(A) = \max_{\rho_{BC},\{B_y\},\{C_z\}} \sum_{y,z} A_{yz}\, \tr\!\left[ \rho_{BC}(B_y\otimes C_z) \right],
\end{equation}
where $B_y$ and $C_z$ are dichotomic observables. Since Bob's and Charlie's output Hilbert spaces are not dimension-bounded, the different message blocks can be optimized independently. Indeed, one can realize independent optimal strategies by placing the states associated with different messages in orthogonal sectors and using block-diagonal measurements. Therefore, for a fixed encoding $f$,
\begin{equation}
\beta_{\mathcal{CQ}}(f) = \sum_{m\in\mathcal M} \beta_Q\!\left(A^{(m\vert f)}\right).
\end{equation}
Maximizing over all deterministic encodings gives
\begin{equation}
\beta_{\mathcal{CQ}} = \max_{f:\mathcal X\to\mathcal M} \sum_{m\in\mathcal M} \beta_Q\!\left(A^{(m\vert f)}\right).
\label{eq:CQ_bound_exact}
\end{equation}

The values $\beta_Q(A)$ are upper bounded using the standard Navascués-Pironio-Acín (NPA) hierarchy~\cite{navascues2007bounding}. For a given level of the hierarchy, the optimization reduces to a semidefinite program maximizing the objective function $\sum_{y,z} A_{yz}\, \Gamma_{B_y,C_z}$ over a moment matrix $\Gamma$, subject to the usual constraints of positivity, normalization, and commutation of Bob's and Charlie's observables. We denote this maximum value by $\beta_{\mathrm{NPA}}(A)$.

In the numerical results below, we use the NPA hierarchy at levels $\ell=1,2,3$. The values change only negligibly between consecutive levels for the witnesses considered here, so we present the level-$3$ bounds. The resulting upper bound on the classical--quantum value is
\begin{equation}
\beta_{\mathcal{CQ}} \le \max_{f:\mathcal X\to\mathcal M} \sum_{m\in\mathcal M} \beta_{\mathrm{NPA}}^{(3)} \!\left(A^{(m\vert f)}\right).
\label{eq:CQ_bound_NPA}
\end{equation}

Applying this procedure to the two fully classical witnesses in Eqs.~\eqref{eq:CC} and \eqref{eq:CC2}, we obtain
\begin{align}
&W_{\mathcal{CC}}^{(1)} \underset{\mathcal{CC}}{\leq} 8 \underset{\mathcal{CQ}}{\leq} 9.2376,\\
&W_{\mathcal{CC}}^{(2)} \underset{\mathcal{CC}}{\leq} 6 \underset{\mathcal{CQ}}{\leq} 8.4853 .
\label{eq:CQ_bound_CC}
\end{align}
These are NPA upper bounds on the corresponding $\mathcal{CQ}$ values.

\subsection{Analytic classical-message bounds for two-body witnesses}
\label{subsec:analytic-classical-message}

For two-body correlator witnesses of the form~\eqref{eq:General_two_Corr_ineq} the optimal value over any classical-message model admits a simple closed-form expression. 

Following the discussion in the previous section, let $R\in\{\mathcal C,\mathcal Q,\mathcal{NS}\}$ denote whether the message-conditioned $BC$ block is local, quantum, or nonsignalling, and let $\beta_R(A)$ be the maximum of $\sum_{y,z}A_{yz}\langle B_yC_z\rangle$ over the resource $R$. The exact classical-message bound is then
\begin{equation}
\beta_{\mathcal C R}(W_\alpha) = \max_{f:\mathcal X\to\mathcal M} \sum_{m\in\mathcal M} \beta_R\!\left(A^{(m\vert f)}\right).
\label{eq:classical_message_two_body_bound}
\end{equation}
The inequality $\beta_{\mathcal C R}(W_\alpha)\le\max_f\sum_m\beta_R(A^{(m\vert f)})$ is immediate, since once the encoding is fixed each message contributes through an independent $BC$ block, whose individual maximum is by definition $\beta_R(A^{(m\vert f)})$. The reverse inequality requires showing that the per-message optima can be attained simultaneously inside a single $BC$ realization. This is where the absence of dimension constraints on Bob and Charlie enters. Given optimal strategies $(\rho_m^{BC},\{B_y^{(m)}\},\{C_z^{(m)}\})$ for each message, place them on orthogonal sectors $\mathcal H_B=\bigoplus_m\mathcal H_B^{(m)}$, $\mathcal H_C=\bigoplus_m\mathcal H_C^{(m)}$, and define the block-diagonal observables $B_y=\bigoplus_m B_y^{(m)}$, $C_z=\bigoplus_m C_z^{(m)}$. The broadcaster then prepares $\rho_m^{BC}$ embedded in the $m$-th sector, and the resulting strategy attains $\sum_m\beta_R(A^{(m\vert f)})$. The same argument applies for $R=\mathcal{NS}$ because the nonsignalling polytope is closed under direct sums.

For the binary-input $BC$ blocks considered here, each $\beta_R$ admits a compact form,
\begin{align}
\beta_{\mathcal C}(A)
&=
\max_{s_y,t_z\in\{\pm1\}}
\sum_{y,z}A_{yz}\,s_y t_z,
\label{eq:betaC_block}\\
\beta_{\mathcal Q}(A)
&=
\max_{\{u_y\},\{v_z\}}
\sum_{y,z}A_{yz}\,u_y\cdot v_z,
\label{eq:betaQ_block}\\
\beta_{\mathcal{NS}}(A)
&=
\sum_{y,z}|A_{yz}|.
\label{eq:betaNS_block}
\end{align}
In Eq.~\eqref{eq:betaQ_block} the maximization is over real unit vectors. This is the standard Tsirelson vector form for bipartite correlator expressions and can be solved as a semidefinite program. For the $2\times2$ case,
\begin{equation}
A=\begin{pmatrix}a&b\\ c&d
\end{pmatrix},
\end{equation}
the quantum value reduces to a one-parameter optimization,
\begin{equation}
\beta_{\mathcal Q}(A)=
\max_{\gamma\in[-1,1]}
\qty[
\sqrt{a^2+c^2+2ac\,\gamma}
+
\sqrt{b^2+d^2+2bd\,\gamma}
],
\label{eq:betaQ_2by2_closed}
\end{equation}
obtained by optimizing over Charlie's vectors and setting $\gamma=u_0\cdot u_1$.

Applying Eq.~\eqref{eq:classical_message_two_body_bound} to the witnesses used below gives the exact bounds listed in Table~\ref{tab:classical_message_bounds}. For $W_{\mathcal{CC}}^{(1)}$ in Eq.~\eqref{eq:CC}, the optimal encoding is $f(0)=f(1)=0$ and $f(2)=1$, which yields the effective coefficient matrices
\begin{equation}
A^{(0)}=\begin{pmatrix}3&-1\\ -1&-1\end{pmatrix},
\qquad
A^{(1)}=\begin{pmatrix}-3&1\\ 1&1\end{pmatrix}.
\end{equation}
Both reduce to the same one-parameter expression $\sqrt{10-6\gamma}+\sqrt{2+2\gamma}$ from Eq.~\eqref{eq:betaQ_2by2_closed}, whose maximum is attained at $\gamma=-\tfrac13$ and equals $\tfrac{8}{\sqrt3}$, giving $\beta_{\mathcal{CQ}}(W_{\mathcal{CC}}^{(1)})=\tfrac{16}{\sqrt3}$. For $W_{\mathcal{CC}}^{(2)}$ in Eq.~\eqref{eq:CC2}, the optimal encoding instead groups $x=0,1$ together and isolates $x=2$, giving
\begin{equation}
A^{(0)}=\begin{pmatrix}2&2\\ -2&2\end{pmatrix},
\qquad
A^{(1)}=\begin{pmatrix}-1&-1\\ 1&-1\end{pmatrix},
\end{equation}
whose quantum correlator values are $4\sqrt2$ and $2\sqrt2$, so that $\beta_{\mathcal{CQ}}(W_{\mathcal{CC}}^{(2)})=6\sqrt2$.

\begin{table}[t]
    \centering
    \renewcommand{\arraystretch}{1.25}
    \caption{Exact classical-message bounds obtained from Eq.~\eqref{eq:classical_message_two_body_bound} for the two-body witnesses considered in this work. The three columns correspond to local, quantum, and nonsignalling resources in the message-conditioned $BC$ block.}
    \label{tab:classical_message_bounds}
    \begin{tabular*}{\columnwidth}{@{\extracolsep{\fill}}lccc@{}}
        \hline\hline
        Witness & $\beta_{\mathcal{CC}}$ & $\beta_{\mathcal{CQ}}$ & $\beta_{\mathcal{CNS}}$ \\
        \hline
        $W_{\mathcal{CC}}^{(1)}$ & $8$ & $16/\sqrt3\simeq9.2376$ & $12$ \\
        $W_{\mathcal{CC}}^{(2)}$ & $6$ & $6\sqrt2\simeq8.4853$ & $12$ \\
        $\mathcal W_{\mathcal{CNS}}^{(3)}$ & $4$ & $4$ & $4$ \\
        \hline\hline
    \end{tabular*}
\end{table}

\subsection{An SDP approximation for the quantum-quantum broadcasting model}
\label{subsec:qqineqs}

We now introduce an SDP relaxation to upper bound linear witnesses in the fully quantum broadcasting model. In this scenario we fix the dimension of Alice's quantum system, while allowing Bob and Charlie's Hilbert spaces to have arbitrary dimension. To deal with this, the method combines ideas from the NPA hierarchy~\cite{navascues2007bounding,navascues2008convergent} and finite-dimensional moment relaxations~\cite{navascues2015bounding,navascues2015characterizing,pironio2010convergent}. A similar combination of finite-dimensional and unrestricted components was previously considered for entanglement-assisted PAM scenarios~\cite{tavakoli2021correlations}, and here we likewise parametrize the finite-dimensional part explicitly before constructing the relaxation.

In the $\mathcal{QQ}$ model, Alice sends a $d$-dimensional quantum state through a broadcasting channel, as in Eq.~\eqref{eq:born_broadcasting}. Since the witnesses we consider are linear in each preparation and the set of states is convex, it is enough to consider pure states $\rho_x=\ket{\psi_x}\bra{\psi_x}$. By Stinespring dilation of the channel and Naimark dilation of the measurements~\cite{Stinespring1955,naimark1940spectral}, we write
\begin{equation}
p(b,c\vert x,y,z) = \bra{\psi_x} U^\dagger P^B_{b\vert y}P^C_{c\vert z}U \ket{\psi_x}.
    \label{eq:qq_stinespring_model}
\end{equation}
where $U:\mathbb C^d\to\mathcal K$ is an isometry satisfying $U^\dagger U=\iden_d$, and $P^B_{b\vert y}$ and $P^C_{c\vert z}$ are projectors acting on the dilation space $\mathcal K$. We relax the tensor product structure by imposing the commuting-operator relations $[P^B_{b\vert y},P^C_{c\vert z}]=0$. 

Now, let us choose an input basis $\{\ket{i}\}_{i=1}^d$, define $\ket{e_i}=U\ket{i}$, and write $\ket{\psi_x}=\sum_i\alpha_{x,i}\ket{i}$. Therefore the probability reads,
\begin{equation}
p(b,c\vert x,y,z)=\sum_{i,j=1}^d\alpha_{x,i}^*\alpha_{x,j}\bra{e_i}P^B_{b\vert y}P^C_{c\vert z}\ket{e_j}.
    \label{eq:qq_basis_expansion}
\end{equation}
Notice that the optimization  variables are  $\alpha_{x,i}^*\alpha_{x,j}$ and the moments  $\bra{e_i}u^\dagger v\ket{e_j}$ for $u$ and $v$ projectors. Then, since any witness would be quadratic on these variables, we relax the problem as follows. 

For binary outcomes we use the observables $B_y=P^B_{0\vert y}-P^B_{1\vert y}$ and $C_z=P^C_{0\vert z}-P^C_{1\vert z}$. The measurement algebra is generated by
\begin{equation}
B_y^2=C_z^2=\iden \quad \text{and} \quad [B_y,C_z]=0 .
    \label{eq:qq_binary_relations}
\end{equation}
Equivalently, in projector notation one has $(P^B_{b|y})^\dagger=P^B_{b|y}$, $P^B_{b|y}P^B_{b'|y}=\delta_{bb'}P^B_{b|y}$, and $\sum_bP^B_{b|y}=\iden$, with analogous relations for Charlie, together with $[P^B_{b|y},P^C_{c|z}]=0$. For binary outcomes, however, these relations are encoded more compactly by Eq.~\eqref{eq:qq_binary_relations}, and we use the observable representation throughout the hierarchy. We denote by $\mathcal S_w$ the set of reduced words in the letters $B_y$ and $C_z$ with length at most $w$. Reduction is always performed modulo Eq.~\eqref{eq:qq_binary_relations}. Thus, $B_yB_yC_z$ is reduced to $C_z$, and $C_zB_y$ is identified with $B_yC_z$. The parameter $w$ is the measurement-word level of the relaxation.

The variables  $\alpha_{x,j}$ are just complex numbers and therefore they commute. We denote by $\mathcal A_q$ the set of monomials in these variables with total degree at most $q$. Thus, $q$ is the preparation-word level of relaxation. The normalization of each preparation is encoded by
\begin{equation}
g_x(\alpha,\alpha^*):=\sum_{i=1}^d \alpha_{x,i}^*\alpha_{x,i}-1=0, \qquad \forall x .
\label{eq:qq_prep_normalization}
\end{equation}

At relaxation order $(q,w)$, the formal vectors kept in the moment matrix are $\mu u\ket{e_i}$, with $\mu\in\mathcal A_q$, $u\in\mathcal S_w$, and $i=1,\ldots,d$. Equivalently, for $\mu,\kappa\in\mathcal A_q$ and $u,v\in\mathcal S_w$, we introduce symbols $\eta^{ij}_{u,v}$ representing $\bra{e_i}u^\dagger v\ket{e_j}$ and a linear functional $L_{q,w}$ acting on moments of the form $L_{q,w}\!\left(\mu^*\kappa\,\eta^{ij}_{u,v}\right)$. For a  quantum realization, its action is defined as
\begin{equation}
    L_{q,w}\!\left(\mu^*\kappa\,\eta^{ij}_{u,v}\right)
    =\mu^*(\alpha,\alpha^*)\, \kappa(\alpha,\alpha^*)\bra{e_i}u^\dagger v\ket{e_j}.
    \label{eq:qq_true_functional}
\end{equation}
This makes Eq.~\eqref{eq:qq_basis_expansion} linear in the SDP variables, because the products between preparation amplitudes and projectors moments are treated as single lifted moments. 

The lifted moment matrix $\mathcal M$ is indexed by triples $(\mu,i,u)$, with $\mu\in\mathcal A_q$, $u\in\mathcal S_w$, and $i=1,\ldots,d$. Its entries are
\begin{equation}
    \mathcal M_{q,w}[(\mu,i,u),(\kappa,j,v)] =  L_{q,w}\!\left( \mu^*\kappa\,\eta^{ij}_{u,v} \right).
    \label{eq:qq_lifted_moment_matrix}
\end{equation}
By definition, for quantum realizations
\begin{equation}\label{eq:qq_lifted_psd}
\mathcal M_{q,w}\succeq0,    
\end{equation}
and the normalization follows from the condition $\langle e_i | e_j \rangle =\langle i |U^\dagger U| j \rangle =  \delta_{ij}$, resulting in 
\begin{equation}\label{eq:qq_frame_normalization}
 \mathcal M_{q,w}[(1, i, \iden), (1, j, \iden)]  = \delta_{ij}.
\end{equation}
Since the vectors $\{\ket{e_i }\}$ form an orthonormal basis, to ensure the SDP solver respects this geometry across all matrix blocks, we also impose the scaled orthonormality condition
\begin{equation}
    \mathcal M_{q,w}[(\mu, i, \iden), (\kappa, j, \iden)]  = \delta_{ij}  \mathcal  M_{q,w}[(\mu, 1, \iden), (\kappa, 1, \iden)] .
    \label{eq:qq_scaled_orthonormality}
\end{equation}

The preparation constraints are imposed through localizing equations. Since $g_x$ has degree two, we multiply it only by monomials whose total degree still fits inside the order $2q$. Thus, for $\mu,\kappa\in\mathcal A_{q-1}$ and for all word pairs kept at level $w$, we impose
\begin{equation}\label{eq:qq_prep_localizing}
    L_{q,w}\!\left( \mu^*\kappa\,g_x\,\eta^{ij}_{u,v} \right) = 0
\end{equation}
At $q=1$, this already enforces the basic normalization $L_{q,w}(g_x\,\eta^{ij}_{u,v})=0$. Larger values of $q$ add the same normalization multiplied by higher-degree preparation monomials.

Since the preparation variables commute, each moment depends only on the commutative monomial $\mu^*\kappa$, not on the particular pair $(\mu,\kappa)$ used to generate it. Thus, if $\mu^*\kappa=r^*s$ as monomials in $\alpha_{x,i}$ and $\alpha_{x,i}^*$, then, for all $u,v\in\mathcal S_w$, we impose
\begin{equation}
L_{q,w}\!\left( \mu^*\kappa\,\eta^{ij}_{u,v} \right) = L_{q,w}\!\left( r^*s\,\eta^{ij}_{u,v} \right).
    \label{eq:qq_scalar_constraints}
\end{equation}
Equivalently, in the implementation, moment variables are indexed by the reduced scalar monomial $\mu^*\kappa$, rather than by the ordered pair $(\mu,\kappa)$.

The measurement constraints are imposed by word reduction. If $u^\dagger v$ and $r^\dagger s$ reduce to the same word modulo Eq.~\eqref{eq:qq_binary_relations}, then, for all $\mu,\kappa\in\mathcal A_q$, we require
\begin{equation}
L_{q,w}\!\left( \mu^*\kappa\,\eta^{ij}_{u,v} \right) = L_{q,w}\!\left( \mu^*\kappa\,\eta^{ij}_{r,s} \right).
    \label{eq:qq_word_constraints}
\end{equation}
For example, taking $u=B_y$, $v=B_yC_z$, $r=C_z$, and $s=\iden$, one has $u^\dagger v=B_y(B_yC_z)=C_z=r^\dagger s$. Hence the two formal moments $\eta^{ij}_{B_y,B_yC_z}$ and $\eta^{ij}_{C_z,\iden}$ must be identified. Similarly, $C_zB_y$ and $B_yC_z$ are identified by Bob--Charlie commutation.

For a general linear witness written in correlator form, we denote by $W_{q,w}(L_{q,w})$ the lifted expression obtained by replacing the correlators with their corresponding lifted moments:
\begin{align}
    W_{q,w}(L_{q,w}) &=  \sum_{x,y,i,j} \omega^B_{xy} L_{q,w} \left(\alpha_{x,i}^*\alpha_{x,j}\eta^{ij}_{\iden,B_y}\right)
    \nonumber\\
    &\quad + \sum_{x,z,i,j}\omega^C_{xz} L_{q,w} \left(\alpha_{x,i}^*\alpha_{x,j}\eta^{ij}_{\iden,C_z}\right)
    \nonumber\\
    &\quad + \sum_{x,y,z,i,j} \omega^{BC}_{xyz} L_{q,w} \left(\alpha_{x,i}^*\alpha_{x,j}
        \eta^{ij}_{B_y,C_z}\right).
    \label{eq:lifted_QQ_witness}
\end{align}
The order-$(q,w)$ relaxation is then
\begin{align}
    \beta_{\mathcal{QQ}}^{(q,w)} = \max_{L_{q,w}}\quad &  W_{q,w}(L_{q,w})  \nonumber\\
    \mathrm{s.t.}\quad
    &
    \text{Eqs.~\eqref{eq:qq_lifted_psd}, \eqref{eq:qq_frame_normalization},\eqref{eq:qq_scaled_orthonormality},\eqref{eq:qq_prep_localizing},\eqref{eq:qq_scalar_constraints}, and \eqref{eq:qq_word_constraints}.}
    \label{eq:qq_sdp_relaxation}
\end{align}
Every physical commuting-operator $\mathcal{QQ}$ realization defines a feasible $L_{q,w}$ through Eq.~\eqref{eq:qq_true_functional}, and evaluating $W_{q,w}(L_{q,w})$ on these moments gives back the same witness value. The SDP therefore optimizes over a set containing all commuting-operator realizations, and hence
$\beta_{\mathcal{QQ}}^{(q,w)}\geq\beta_{\mathcal{QQ}}^{\mathrm{com}}(d)$,
where $\beta_{\mathcal{QQ}}^{\mathrm{com}}(d)$ denotes the optimal value in the commuting-operator model. Since every tensor-product realization is also a commuting-operator realization, we further have
$\beta_{\mathcal{QQ}}^{\mathrm{com}}(d)\geq\beta_{\mathcal{QQ}}^{\otimes}(d)$,
with $\beta_{\mathcal{QQ}}^{\otimes}(d)$ denoting the corresponding tensor-product optimum. We keep these two quantities distinct, since tensor-product and commuting-operator correlation models are not equivalent in general \cite{junge2011connes,ji2021mip}.

Increasing $q$ or $w$ keeps all lower-level moments and adds further consistency conditions on the preparation and measurement variables. Therefore the feasible sets are nested, and the hierarchy is monotone non-increasing. In particular, for $q'\geq q$ and $w'\geq w$, $\beta_{\mathcal{QQ}}^{(q',w')} \leq \beta_{\mathcal{QQ}}^{(q,w)}$. In Appendix~\ref{app:convergence}, following Refs.~\cite{navascues2008convergent,pironio2010convergent}, we show that the SDP behavior sets form a nested sequence of outer approximations converging to the closed convex hull of the commuting-operator $\mathcal{QQ}$ set, and consequently that the corresponding upper bounds converge to $\beta_{\mathcal{QQ}}^{\mathrm{com}}(d)$ as both relaxation orders tend to infinity.

As a numerical check, we evaluated the SDP relaxation at the first two levels, $(q,w)=(1,1)$ and $(q,w)=(1,2)$. The results are collected in Table~\ref{tab:qq_witness_values}. The optimal setups achieving these values are reported in Appendix~\ref{Optimal_configurations}. For the two fully classical witnesses in Eqs.~\eqref{eq:CC} and \eqref{eq:CC2}, the SDP upper bounds are already nearly stable at the first level. The same is true for the witness in Eq.~\eqref{eq:WNS3}.

\begin{table}[t]
    \centering
    \renewcommand{\arraystretch}{1.3}
    \caption{SDP upper bounds returned by the lifted relaxation for the witnesses considered here. The second column gives the reference bound of the inequality. For the first two rows this is the fully classical bound, while for the last row it is the classical--$\mathcal{NS}$ bound. The optimal setups achieving these values are reported in Appendix~\ref{Optimal_configurations}.}
    \label{tab:qq_witness_values}
    \begin{tabular*}{\columnwidth}{@{\extracolsep{\fill}}cccc@{}}
        \hline\hline
        Witness & Reference bound & $(1,1)$ \& $(1,2)$-level & Conjectured $\mathcal{QQ}$ \\
        \hline
        $W_{\mathcal{CC}}^{(1)}$ & $8$ & $12.944272$ & $4+4\sqrt{5}$  \\
        $W_{\mathcal{CC}}^{(2)}$ & $6$ & $10.828427$ & $8+2\sqrt{2}$  \\
        $\mathcal W_{\mathcal{CNS}}^{(3)}$ & $4$ & $5.196152$ & $3\sqrt{3}$  \\
        \hline\hline
    \end{tabular*}
\end{table}

\section{Input-free measurements}
\label{sec:A NO-GO THEOREM}

The simplest prepare-and-broadcast configurations are those in which Bob and Charlie have no independent measurement inputs. In this regime, they each perform one fixed measurement, and the observed behavior is the conditional distribution $p(b,c\vert x)$. The main simplification is that the part of the experiment after Alice sends the system, namely the broadcasting channel followed by Bob's and Charlie's fixed measurements, can be represented by a single effective POVM acting on the $d$-dimensional system sent by Alice.

This observation connects the input-free PAB scenario directly to the Frenkel--Weiner theorem~\cite{frenkel2015classical}: the statistics of a single POVM on a $d$-dimensional quantum system can be simulated, with shared randomness, by a classical message with at most $d$ values. We use this to prove that, in the input-free regime and without pre-shared entanglement between Alice and the broadcasting device, quantum messages give no advantage over $d$-valued classical messages.

We then discuss the entanglement-assisted input-free case. Such entanglement lies outside the causal model of the collapse theorem and can produce violations. However, because there are no receiver inputs, the joint output $o=(b,c)$ can be regarded as a single four-valued outcome. Consequently, the entanglement-assisted violation in the minimal binary-output input-free PAB scenario is not a genuinely multipartite separation. It is the known four-output input-free PAM advantage, implemented in a distributed form.

\subsection{Input-free measurements with shared randomness}
\label{subsec:input-free-shared-randomness}

Throughout this section, the preparation alphabet $\mathcal X$ and the output alphabets $\mathcal B$ and $\mathcal C$ are finite. We denote by $\mathsf{\mathcal{CC}}^{\mathrm{if}}_d$ the input-free classical set, consisting of behaviors of the form
\begin{equation}
\label{eq:if-classical}
p(b,c\vert x) = \sum_{\lambda,m} p(\lambda)\,D^A_{\lambda}(m\vert x)\,D^B_{\lambda}(b\vert m)\,D^C_{\lambda}(c\vert m),
\end{equation}
where $|\mathcal M|\le d$, the response functions may be taken deterministic, and $\lambda$ is shared randomness independent of $x$.

We denote by $\mathsf{\mathcal{QQ}}^{\mathrm{if}}_d$ the convex hull of the quantum behaviors
\begin{equation}
\label{eq:if-quantum}
p(b,c\vert x) = \operatorname{tr}\!\left[\mathcal T(\rho_x)\left(M^B_b\otimes M^C_c\right)\right],
\end{equation}
where $\rho_x\in\mathcal D(\mathbb C^d)$, $\mathcal T:\mathcal L(\mathbb C^d)\to\mathcal L(\mathcal H_B\otimes\mathcal H_C)$ is an arbitrary quantum channel, and $\{M^B_b\}_b$ and $\{M^C_c\}_c$ are arbitrary finite-outcome POVMs. The dimensions of $\mathcal H_B$ and $\mathcal H_C$ are unrestricted. The convex hull accounts for shared randomness among the devices. The next theorem formalizes this effective-POVM reduction.

\begin{theorem}[Input-free collapse]
\label{thm:if-collapse}
For every finite preparation alphabet and finite output alphabets,
\begin{equation}
\label{eq:if-collapse}
\mathsf{\mathcal{QQ}}^{\mathrm{if}}_d = \mathsf{\mathcal{CC}}^{\mathrm{if}}_d.
\end{equation}
Equivalently, every input-free quantum PAB behavior with a $d$-dimensional input system admits a classical simulation using one common message of cardinality at most $d$, assisted by shared randomness.
\end{theorem}

\begin{proof}
Since $\mathsf{\mathcal{QQ}}^{\mathrm{if}}_d$ is defined as the convex hull of single quantum realizations, and since $\mathsf{\mathcal{CC}}^{\mathrm{if}}_d$ is convex, it is sufficient to show that every behavior of the form Eq.~\eqref{eq:if-quantum} belongs to $\mathsf{\mathcal{CC}}^{\mathrm{if}}_d$.

Let $\mathcal T^*$ denote the adjoint of the broadcasting channel, defined by
\begin{equation}
\label{eq:dual-channel}
\operatorname{tr}[\mathcal T(\rho)A] = \operatorname{tr}[\rho\,\mathcal T^*(A)]
\end{equation}
for every input operator $\rho$ and output operator $A$. Define
\begin{equation}
\label{eq:effective-povm}
E_{bc} := \mathcal T^*\!\left(M^B_b\otimes M^C_c\right).
\end{equation}
Because $\mathcal T$ is completely positive and trace preserving, $\mathcal T^*$ is completely positive and unital. Hence $E_{bc}\succeq0$, and
\begin{equation}
\label{eq:effective-povm-normalization}
\sum_{b,c}E_{bc} = \mathcal T^*\!\left[\left(\sum_b M^B_b\right)\otimes\left(\sum_c M^C_c\right)\right] = \mathcal T^*(\iden_B\otimes\iden_C) = \iden_d.
\end{equation}
Therefore $\{E_{bc}\}_{b,c}$ is a POVM on $\mathbb C^d$. Moreover,
\begin{equation}
\label{eq:collapsed-statistics}
p(b,c\vert x) = \operatorname{tr}(\rho_xE_{bc}).
\end{equation}
Thus the broadcasting channel followed by the two fixed local measurements is equivalent to a single POVM on Alice's $d$-dimensional input space, with the pair $(b,c)$ regarded as one joint outcome.

By the Frenkel--Weiner theorem~\cite{frenkel2015classical}, the statistics of a single POVM on a $d$-dimensional system can be reproduced, with shared randomness, by a classical message with at most $d$ values. Applying the theorem to the outcome alphabet $\mathcal B\times\mathcal C$, Eq.~\eqref{eq:collapsed-statistics} admits a decomposition
\begin{equation}
\label{eq:fw-decomposition}
p(b,c\vert x) = \sum_{\lambda,m} p(\lambda)\,D^A_{\lambda}(m\vert x)\,q_{\lambda}(b,c\vert m), \qquad |\mathcal M|\le d.
\end{equation}
For each fixed $\lambda$, the conditional distribution $q_{\lambda}(b,c\vert m)$ is a classical channel from $m$ to the joint outcome $(b,c)$. Since all alphabets are finite, this channel is a convex combination of deterministic maps $g_\nu:m\mapsto\bigl(h^B_\nu(m),h^C_\nu(m)\bigr)$, with weights attached to whole functions $g_\nu$, not separately to each value of $m$. Absorbing the additional variable $\nu$ into the shared variable $\lambda$, the joint decoder factors into local deterministic responses $D^B_{\lambda}(b\vert m)D^C_{\lambda}(c\vert m)$, and Eq.~\eqref{eq:fw-decomposition} takes exactly the form of Eq.~\eqref{eq:if-classical}. This proves $\mathsf{\mathcal{QQ}}^{\mathrm{if}}_d\subseteq\mathsf{\mathcal{CC}}^{\mathrm{if}}_d$.

For the converse inclusion, it is enough to realize every deterministic classical strategy. Such a strategy is specified by functions $m=f(x)$, $b=h^B(m)$, and $c=h^C(m)$. Alice encodes the message into the orthogonal state $\rho_x=|f(x)\rangle\langle f(x)|$ of $\mathbb C^d$. The broadcaster applies the measure-and-prepare channel
\begin{equation}
\label{eq:if-classical-channel}
\mathcal T(\rho) = \sum_{m=0}^{d-1}\langle m|\rho|m\rangle\,|h^B(m)\rangle\langle h^B(m)|\otimes |h^C(m)\rangle\langle h^C(m)|.
\end{equation}
Local measurements in the corresponding output bases reproduce $b=h^B(f(x))$ and $c=h^C(f(x))$. Convex mixtures of deterministic strategies are obtained by shared randomness. Hence $\mathsf{\mathcal{CC}}^{\mathrm{if}}_d\subseteq\mathsf{\mathcal{QQ}}^{\mathrm{if}}_d$. The two inclusions prove Eq.~\eqref{eq:if-collapse}.
\end{proof}

As an immediate consequence, every linear input-free witness has the same maximum over the classical and quantum sets. Hence no facet of the input-free classical polytope can be violated in the unassisted PAB model with shared randomness. The same collapse holds for the commuting-operator relaxation used to outer-approximate the fully quantum model. Let $\mathsf{\mathcal{QQ}}^{\mathrm{co,if}}_d$ denote the convex hull of behaviors
\begin{equation}
\label{eq:if-commuting-behavior}
p(b,c\vert x) = \operatorname{tr}\!\left[\rho_x\,U^\dagger P^B_bP^C_cU\right],
\end{equation}
where $U:\mathbb C^d\to\mathcal K$ is an isometry satisfying $U^\dagger U=\iden_d$, and $\{P^B_b\}_b$ and $\{P^C_c\}_c$ are projective resolutions on an arbitrary Hilbert space $\mathcal K$ with $[P^B_b,P^C_c]=0$ for every $b,c$. The finite levels of the lifted SDP relaxation of Sec.~\ref{sec:inequalities} are outer approximations to this commuting-operator model.

\begin{corollary}[Hybrid and commuting-operator collapse]
\label{cor:if-commuting}
In the input-free regime,
\begin{equation}
\label{eq:if-full-collapse}
\mathsf{\mathcal{CC}}^{\mathrm{if}}_d = \mathsf{\mathcal{QC}}^{\mathrm{if}}_d = \mathsf{\mathcal{CQ}}^{\mathrm{if}}_d = \mathsf{\mathcal{CNS}}^{\mathrm{if}}_d = \mathsf{\mathcal{QQ}}^{\mathrm{if}}_d = \mathsf{\mathcal{QQ}}^{\mathrm{co,if}}_d.
\end{equation}
\end{corollary}

\begin{proof}
Consider a commuting-operator behavior of the form Eq.~\eqref{eq:if-commuting-behavior}. Since $P^B_b$ and $P^C_c$ are commuting projectors, $\Pi_{bc}:=P^B_bP^C_c$ is itself a projector. The completeness relations give
\begin{equation}
\label{eq:commuting-projector-normalization}
\sum_{b,c}\Pi_{bc} = \left(\sum_bP^B_b\right)\left(\sum_cP^C_c\right) = \iden_{\mathcal K}.
\end{equation}
Consequently, $E_{bc}:=U^\dagger\Pi_{bc}U$ defines a POVM on $\mathbb C^d$. The same Frenkel--Weiner argument used in Theorem~\ref{thm:if-collapse} gives $\mathsf{\mathcal{QQ}}^{\mathrm{co,if}}_d\subseteq\mathsf{\mathcal{CC}}^{\mathrm{if}}_d$. The converse inclusion holds because the deterministic realization constructed in the proof of Theorem~\ref{thm:if-collapse} is also a commuting-operator realization.

The equality $\mathsf{\mathcal{QQ}}^{\mathrm{if}}_d=\mathsf{\mathcal{CC}}^{\mathrm{if}}_d$ is Theorem~\ref{thm:if-collapse}. The equality $\mathsf{\mathcal{QC}}^{\mathrm{if}}_d=\mathsf{\mathcal{CC}}^{\mathrm{if}}_d$ follows by restricting the general equality $\mathsf{\mathcal{QC}}_d=\mathsf{\mathcal{CC}}_d$ to a single measurement setting for each receiver.

It remains to consider the classical-message hybrid models. In the $\mathsf{\mathcal{CQ}}^{\mathrm{if}}_d$ and $\mathsf{\mathcal{CNS}}^{\mathrm{if}}_d$ models, conditioning on the classical message $m$ and shared variable $\lambda$ leaves only a joint distribution $q_{\lambda}(b,c\vert m)$. Since there are no receiver inputs, there are no nontrivial nonsignalling constraints beyond positivity and normalization. Every finite channel $m\mapsto(b,c)$ is a convex combination of deterministic maps, and the corresponding convex weights can be absorbed into the shared randomness. Thus both hybrid models reduce to Eq.~\eqref{eq:if-classical}. The reverse inclusions are immediate.
\end{proof}

The input-free assumption is essential. If Bob and Charlie receive measurement inputs $y$ and $z$, the single effective POVM is replaced by the family
\begin{equation}
\label{eq:if-input-family}
E_{bc\vert yz} = \mathcal T^*\!\left(M^B_{b\vert y}\otimes M^C_{c\vert z}\right).
\end{equation}
The Frenkel--Weiner theorem can be applied separately to each POVM in this family, but it does not generally provide one common $d$-valued encoding that is simultaneously compatible with every choice of $(y,z)$. The individual simulations may require setting-dependent codebooks, while Alice's preparation device has no access to the receivers' inputs. This obstruction is precisely what allows quantum advantages to reappear when independent receiver inputs are present.

\subsection{Joint-output PAM equivalence and entanglement assistance}
\label{subsec:input-free-pam-equivalence}

The proof above has an important interpretational consequence. In the input-free regime, the observed data do not distinguish a distributed two-receiver readout from a single-receiver readout with joint outcome $o=(b,c)\in\mathcal O:=\mathcal B\times\mathcal C$. In particular, the classical input-free PAB model can be rewritten as
\begin{equation}
\label{eq:if-pam-classical}
p(o\vert x) = \sum_{\lambda,m}p(\lambda)\,D^A_{\lambda}(m\vert x)\,D^O_{\lambda}(o\vert m),
\end{equation}
which is exactly the standard input-free PAM model with output alphabet $\mathcal O$. Conversely, every deterministic decoder $o=g_{\lambda}(m)$ specifies deterministic values $b=h^B_{\lambda}(m)$ and $c=h^C_{\lambda}(m)$, so the input-free PAB and joint-output PAM descriptions generate the same classical correlation set.

Consider the minimal binary-output case $|\mathcal X|=3$, $d=2$, and $|\mathcal B|=|\mathcal C|=2$. Then the joint output alphabet has four values. One representative nontrivial facet of the corresponding classical polytope is
\begin{equation}
\label{eq:input-free-Wif}
W_{\mathrm{if}} :=  2\left[ p(00|0)+p(10|1)+p(01|2) \right] + \sum_{x=0}^{2}p(11|x) \le 4 .
\end{equation}
Identifying the four joint outcomes as $00\mapsto0$, $10\mapsto1$, $01\mapsto2$, and $11\mapsto\perp$, this inequality becomes
\begin{equation}
\label{eq:input-free-S1}
W_{\mathrm{if}} := \frac14 \sum_{x=0}^{2} \left[2p(o=x|x)+p(o=\perp|x) \right] \le 1 .
\end{equation}
Thus Eq.~\eqref{eq:input-free-Wif} is precisely the four-output input-free PAM witness of Ref.~\cite{Svegborn2026}, written with the joint outcome $o=(b,c)$.

We now explain how the entanglement-assisted realization fits this equivalence. Suppose Alice and the receiving device share $\omega_{AR}\in\mathcal D(\mathcal H_A\otimes\mathcal H_R)$ before the protocol. Upon receiving $x$, Alice measures $\{M_{a\vert x}\}_a$ on $A$, sends the classical message $a\in\{0,1\}$, and remotely prepares the subnormalized state
\begin{equation}
\label{eq:ea-assemblage}
\sigma^R_{a\vert x} = \operatorname{tr}_A\!\left[(M_{a\vert x}\otimes\iden_R)\omega_{AR}\right]
\end{equation}
on $R$. The nonsignalling condition gives $\sum_a\sigma^R_{a\vert x}=\rho_R$ independently of $x$, but the pair $(a,\sigma^R_{a\vert x})$ can carry more information about $x$ than the classical bit $a$ alone.

In the corresponding input-free PAM protocol, the receiver uses a message-dependent four-outcome POVM $\{N_{o\vert a}\}_{o\in\mathcal O}$, and the behavior is
\begin{equation}
\label{eq:input-free-ea-pam}
p_{\mathrm{EA}}(o\vert x) = \sum_{a=0}^{1}\operatorname{tr}\!\left[N_{o\vert a}\sigma^R_{a\vert x}\right].
\end{equation}
The same behavior can be implemented in the distributed PAB readout. Since Bob's and Charlie's output spaces are unrestricted, each four-outcome POVM $\{N_{o\vert a}\}_{o}$ admits a Naimark realization with fixed binary local measurements and an $a$-dependent isometry:
\begin{equation}
\label{eq:ea-pab-effective-povm}
N_{bc\vert a} = U_a^\dagger\left(M^B_b\otimes M^C_c\right)U_a, \qquad o=(b,c).
\end{equation}
Here $U_a:\mathcal H_R\to\mathcal H_B\otimes\mathcal H_C$ is an isometry depending on the communicated message $a$, while $\{M^B_b\}_b$ and $\{M^C_c\}_c$ are fixed binary local measurements. Operationally, the four POVM outcomes are encoded into the two binary registers $b$ and $c$.

Using the two-qubit entanglement-assisted strategy of Ref.~\cite{Svegborn2026} for the witness $ W_{\mathrm{if}}$, one obtains \begin{equation}
\label{eq:S1-ea-value}
W_{\mathrm{if}}^{\mathrm{EA}} = \frac{9+2\sqrt3}{12} \approx 1.0387.
\end{equation}
This violation should not be interpreted as a genuinely multipartite advantage of the unrestricted input-free PAB scenario. In the absence of receiver inputs, the joint output $o=(b,c)$ is a single four-valued outcome, and the witness is affinely equivalent to the known four-output input-free PAM witness. The result is therefore a distributed implementation of the entanglement-assisted PAM advantage, not a new separation between input-free PAB and input-free PAM.

\section{Activation of nonclassicality in the prepare and broadcast scenario}
\label{sec:quantum-broadcasting}

The no-go theorem of Sec.~\ref{sec:A NO-GO THEOREM} shows that independent receiver inputs are necessary for a quantum advantage within the PAB model with shared randomness. We now investigate two ways in which broadcasting can enhance the detection of nonclassicality once these measurement choices are restored. The first concerns genuine activation at the level of the broadcasting channel, while the second concerns the robustness of detecting nonclassicality in Alice's preparations relative to finite families of standard PAM measurements. 

The first form concerns the broadcasting channel itself. A broadcasting channel may be entangling and can therefore generate Bell-nonlocal correlations between Bob and Charlie. However, if enough noise is added, the state shared by Bob and Charlie may no longer violate any Bell inequality, even if it remains entangled. A standard way to model this situation is to consider a noisy two-qubit state of the form $v\rho + (1-v)\tfrac{\mathbb{I}}{4}$. For the maximally entangled state $\rho=\ket{\Phi^+}\bra{\Phi^+}$, with $\ket{\Phi^+}=2^{-\nicefrac12}(\ket{00}+\ket{11})$, the CHSH threshold is $v=\tfrac{1}{\sqrt2}\simeq0.7071$. More generally, for correlation Bell tests with projective measurements, the Werner state does not violate any Bell inequality whenever $v\le \tfrac{1}{G_3}$, where $\tfrac{1}{G_3}\simeq0.697$ where $G_3$ is Grothendieck's constant of order three. It is also known that this hidden nonclassicality can be activated when multiple copies of the noisy state are available \cite{palazuelos2012superactivation,cavalcanti2011quantum}.

In the single-copy regime, Ref.~\cite{bowles2021single} showed that combining a Bell scenario with broadcasting can lower the relevant visibility threshold to $v=0.559$. Here we obtain a similar effect, but with a different placement of the noise. In Ref.~\cite{bowles2021single}, the visibility parameter acts on the entangled state initially shared between Alice and Bob, while the broadcasting channel from Bob to Charlie is ideal. In our case, the qubit states prepared by Alice are ideal, and we instead add white noise to the output of the broadcasting channel. We show that, using fully classical PAB inequalities, the channel visibility threshold can be reduced to $v_{\mathrm{ch}}\simeq0.554$. This shows that the PAB scenario can reveal nonclassicality of the broadcasting stage in a regime where standard Bell tests would not detect it.

We next consider nonclassicality at the level of Alice's preparations. Here the goal is to compare the noise robustness of PAB witnesses with that obtained from standard PAM tests applied to the same sets of qubit states. General methods for proving PAM classicality of a set of states are known \cite{de2021general}, but they can give loose bounds on the critical visibility. We therefore use an efficient linear-programming method \cite{renner2023classical} to test the preparations against large finite sets of projective measurements. Since these measurement sets are finite, this comparison does not establish classicality with respect to all possible PAM measurements.

To attribute a PAB violation to the preparation stage, it is important not to compare only with the fully classical bound. A violation of a fully classical PAB inequality could still be explained by a classical preparation message followed by a quantum broadcasting device. Consequently, when the aim is to certify nonclassicality in the preparation stage, the relevant benchmarks are the hybrid bounds $\beta_{\mathcal{CQ}}$ or $\beta_{\mathcal{CNS}}$. A violation above these bounds cannot be attributed only to the broadcasting device.

We now describe the numerical parametrization used in the searches below. Alice prepares pure qubit states $\rho_x = \frac12\left(\iden+\vec r_x\cdot\vec\sigma\right)$, with $\|\vec r_x\|_2=1$, and Bob and Charlie measure the dichotomic observables $B_y=\vec b_y\cdot\vec\sigma$, $C_z=\vec c_z\cdot\vec\sigma$, where $\|\vec b_y\|_2=\|\vec c_z\|_2=1$. The witness value is optimized over the preparation Bloch vectors, the local
measurement directions, and the broadcasting map.

For an inner-point optimization, we parametrize the broadcasting map as a qubit-to-two-qubit isometry. We start from a unitary on $\mathcal H_B\otimes\mathcal H_C\simeq\mathbb C^4$, 
\begin{equation}
U_{\mathrm{BC}}(\mathbf p) =\exp\left[-i\sum_{k=1}^{15}p_k\Lambda_k\right],
\qquad
\mathbf p\in\mathbb R^{15},
\end{equation}
where $\{\Lambda_k\}_{k=1}^{15}$ is a basis of traceless Hermitian generators of $SU(4)$. Fixing an embedding $E:\mathcal H_M\to\mathcal H_B\otimes\mathcal H_C$, the isometry is $U(\mathbf p)=U_{\mathrm{BC}}(\mathbf p)E$ and $U(\mathbf p)^\dagger U(\mathbf p)=\iden_2.$ Thus, for each preparation, the output state shared by Bob and Charlie is $\mathcal T(\rho_x) =U(\mathbf p)\rho_x U(\mathbf p)^\dagger =\tau_x^{BC}$, where $\mathcal T$ now denotes the broadcasting channel defined by the parameterized isometry $U(\mathbf p)$. The correlators are then
\begin{align}
\langle B_y\rangle^{[x]} &= \tr\!\left[ \tau_x^{BC}(B_y\otimes\iden) \right],
\nonumber\\
\langle C_z\rangle^{[x]}
&= \tr\!\left[ \tau_x^{BC}(\iden\otimes C_z) \right],
\nonumber\\
\langle B_yC_z\rangle^{[x]}
&= \tr\!\left[ \tau_x^{BC}(B_y\otimes C_z) \right].
\end{align}
Specifically, given a two-body PAB inequality $W = \sum_{x,y,z} \alpha_{xyz} \langle B_yC_z\rangle^{[x]} \le \beta$, we obtain a quantum lower bound $W^{Q}$  via the maximization
\begin{equation}
W^{Q} = \max_{\mathbf p,\{\vec r_x\},\{\vec b_y\},\{\vec c_z\}}
\sum_{x,y,z} \alpha_{xyz} \langle B_yC_z\rangle^{[x]}.
\label{eq:WCC_quantum_optimization}
\end{equation}

\subsection{Activation of nonclassicality in the broadcasting channel}
\label{subsubsec:three-prep-two-body}
Consider a broadcasting channel generating the noisy output state 
\begin{equation}
\tau_x^{BC}(v_{\mathrm{ch}})= v_{\mathrm{ch}}\,\tau_x^{BC}+(1-v_{\mathrm{ch}})\frac{\iden_4}{4}.
\end{equation}
Since $B_y$ and $C_z$ are traceless, the white-noise contribution vanishes for two-body correlators resulting in the relation
\begin{equation}
\langle B_yC_z\rangle^{[x]}_{v_{\mathrm{ch}}}= v_{\mathrm{ch}} \langle B_yC_z\rangle^{[x]}
\end{equation}
Consequently, $W^{Q}(v_{\mathrm{ch}}) = v_{\mathrm{ch}}\,W^{Q}$. Therefore, for any witness  involving only two-body correlators, the critical channel visibility is
\begin{equation}
v_{\mathrm{ch}}^{\ast}=\frac{\beta}{W^{Q}}.
\label{eq:WCC_channel_visibility}
\end{equation}
The best channel-noise robustness we found comes from the inequality $W_{\mathcal{CC}}^{(2)}$ in Eq.~\eqref{eq:CC2}. The numerical optimization gives $(W_{\mathcal{CC}}^{(2)})^Q\simeq 10.828$. This agrees, up to numerical precision, with the SDP outer approximation reported in Sec.~\ref{subsec:qqineqs}. One set of preparations, measurements, and an isometry achieving this value is given in Appendix~\ref{app:config_wcc2}. Since the corresponding fully classical bound is $\beta_{\mathcal{CC}}=6$, we obtain $v_{\mathrm{ch}}^{\ast} = \tfrac{6}{10.828} \simeq 0.554$. This value is below the Werner-state reference threshold $\tfrac{1}{G_3}\simeq0.697$ and also slightly below the value $0.559$ found in Ref.~\cite{bowles2021single}. These thresholds refer to different noise placements, but the comparison shows that the PAB scenario can reveal broadcasting nonclassicality at low channel visibilities.

\begin{figure}[t]
\centering
\includegraphics{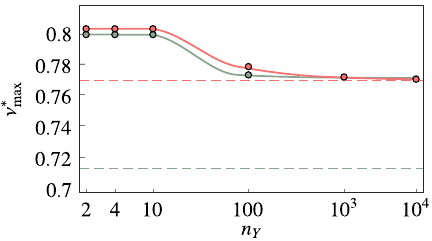}
    \caption{Critical PAM visibility as a function of the number $n_Y$ of Bob's measurement settings. The green curve corresponds to the three preparations optimizing the violation of $W_{\mathcal{CC}}^{(1)}$, whereas the red curve corresponds to the three preparations optimizing the violation of $\mathcal{W}_{\mathcal{CNS}}^{(3)}$. The horizontal dashed lines indicate the corresponding PAB thresholds, $v_{\mathrm{prep}}^{\ast}\simeq 0.713$ for $W_{\mathcal{CC}}^{(1)}$ and $v_{\mathrm{prep}}^{\ast}\simeq 0.769$ for $\mathcal{W}_{\mathcal{CNS}}^{(3)}$. In the first case, the smallest PAM visibility obtained by the linear program remains above the PAB threshold, whereas in the second case the sampled PAM visibility approaches the corresponding PAB threshold as $n_Y$ increases}
    \label{fig:pam_visibility_combined} 
\end{figure}

\subsection{Enhanced detection of nonclassicality in the prepared states}
\label{subsubsec:three-prep-two-body-activation}

In a standard prepare-and-measure scenario, improving noise tolerance may require a large number of measurement settings. Here, we show that a prepare-and-broadcast architecture can achieve lower critical visibilities using only two local measurements per receiver than those obtained within the large finite families of PAM measurements considered below. In our numerical sampling, approaching comparable PAM visibilities requires several thousand measurement settings. This indicates that broadcasting can provide a more measurement-efficient and noise-robust way of detecting preparation nonclassicality.

To attribute such a PAB violation to the preparation stage, however, it is not enough to compare with the fully classical bound. A violation of the fully classical model could still be explained by a classical preparation message together with a quantum broadcasting device. For this reason, the relevant benchmark is the classical-quantum bound, in which Alice's communicated message remains classical, while Bob and Charlie are allowed to share arbitrary bipartite quantum correlations conditioned on that message. Thus, any value above $\beta_{\mathcal{CQ}}$ cannot be reproduced with a classical preparation message, even if the broadcasting stage is quantum.

We now add noise to Alice's preparations. The noisy states are $\rho_x(v_{\mathrm{prep}}) = v_{\mathrm{prep}}\rho_x + (1-v_{\mathrm{prep}})\frac{\iden_2}{2}$, where $v_{\mathrm{prep}}\in[0,1]$. 
Unlike the channel-noise case, preparation noise does not automatically imply that every witness scales as $v_{\mathrm{prep}}$. For a fixed isometry and fixed measurements, a two-body witness behaves as $W(v_{\mathrm{prep}}) = v_{\mathrm{prep}}W + (1-v_{\mathrm{prep}})W_0$, where
\begin{equation}
W_0 = \sum_{x,y,z} \alpha_{xyz} \tr\!\left[ U(\mathbf p)\frac{\iden_2}{2}U(\mathbf p)^\dagger (B_y\otimes C_z) \right].
\end{equation}
Thus, the critical visibility is in general
\begin{equation}
v_{\mathrm{prep}}^{\ast} = \frac{\beta-W_0}{W-W_0}.
\label{eq:prep_visibility_general}
\end{equation}
When the witness is balanced, meaning that $\sum_x \alpha_{xyz}=0$ for all $y,z$, the offset $W_0$ vanishes for any isometry and measurements. In this case the simpler formula $v_{\mathrm{prep}}^{\ast} = \tfrac{\beta}{W}$ applies.

We first consider the fully classical witness $W_{\mathcal{CC}}^{(1)}$ in Eq.~\eqref{eq:CC}. The optimized fully quantum value is $(W_{\mathcal{CC}}^{(1)})^Q \simeq 12.944$, which is larger than the corresponding classical--quantum bound $\beta_{\mathcal{CQ}}\simeq9.2376$. This value also agrees, up to numerical precision, with the SDP upper bound of Sec.~\ref{subsec:qqineqs}. The corresponding optimal strategy is given in Appendix~\ref{app:config_wcc1}.

Since $W_{\mathcal{CC}}^{(1)}$ is balanced, the preparation-noise threshold is $v_{\mathrm{prep}}^{\ast} = \tfrac{\beta_{\mathcal{CQ}}}{(W_{\mathcal{CC}}^{(1)})^Q} \simeq 0.713$. As can be seen in Fig.~\ref{fig:pam_visibility_combined}, the smallest visibility found in the standard PAM test, with up to $10^4$ randomly sampled projective measurements, is approximately $v_{\mathrm{LP}}^{\ast}\simeq0.769$. Thus, the PAB witness detects nonclassicality at a visibility where all sampled PAM tests remain classical, even though this conclusion should be understood within the numerical scope of the finite measurement sampling.

Motivated by this result, we also ask whether a similar robustness advantage can be obtained against the stronger classical--$\mathcal{NS}$ benchmark. For three preparations, we optimize the witness $\mathcal W_{\mathcal{CNS}}^{(3)}$ in Eq.~\eqref{eq:WNS3}. The numerical optimization gives $(\mathcal W_{\mathcal{CNS}}^{(3)})^Q \simeq 3\sqrt{3} \simeq 5.196$, in agreement, up to numerical precision, with the SDP upper bound presented in Sec.~\ref{subsec:qqineqs} (see Appendix~\ref{app:config_wcns3}). Since the classical--$\mathcal{NS}$ bound is $4$, and since this witness is balanced, the corresponding preparation-noise threshold is $v_{\mathrm{prep}}^{\ast} \left(\mathcal W_{\mathcal{CNS}}^{(3)}\right) \simeq 0.769$.  As shown in Fig.~\ref{fig:pam_visibility_combined}, increasing Bob's measurement alphabet improves the PAM noise tolerance, approaching $v^\ast\simeq0.773$ near the end of the sampled range (with up to $10^4$ measurement settings). Within the sampled range, the PAM threshold approaches the corresponding PAB value, and we therefore do not claim a strict activation result in this case. Nevertheless, the broadcasting witness reaches essentially the same robustness using only two settings per receiver, while the PAM test requires several thousand sampled measurements to approach it.

\section{Conclusions and outlook}
\label{secFINAL REMARKS:VIOL}

In this work, we introduced and studied the prepare-and-broadcast scenario with dimension restriction, a natural extension of the standard prepare-and-measure setting in which the communicated system is processed by a broadcasting transformation before being accessed by multiple observers. This framework contains both prepare-and-measure and Bell scenarios as special cases, while allowing one to study how dimension-restricted systems can be distributed among different receivers. We developed a hierarchy of models for the PAB scenario, depending on whether the message sent by Alice and the resource used by the broadcaster are classical, quantum, or nonsignalling. Within this hierarchy, we derived PAB inequalities for the fully classical and classical--nonsignalling models in the simplest binary scenarios and analyzed their violations by quantum broadcasting strategies. We also introduced numerical tools, such as semidefinite-programming relaxations to test classical compatibility, obtain hybrid bounds, and upper-bound the fully quantum broadcasting value of linear witnesses.

Furthermore, we prove analytically that for scenarios with a single measurement setting per party (all of them equipped with shared randomness), the hierarchy of sets collapses, showing the impossibility of non-classicality in this setup. However, by allowing the parties to share entangled systems, we show that violation of prepare-and-broadcast witnesses can be found. 

Our results show that broadcasting can reveal nonclassical features that are not visible in the corresponding simpler scenarios. First, we found that PAB inequalities can detect nonclassicality of a noisy broadcasting channel below standard Bell-inequality thresholds. Second, we studied the detection of nonclassicality at the level of the prepared states. In this case, the relevant comparison is not only with the fully classical model, but also with hybrid models in which the broadcasting device is allowed to be quantum or nonsignalling. We showed that broadcasting inequalities can be more noise-robust than the finite families of standard PAM tests considered for the same preparations, and we provided evidence that this advantage persists even when the preparations are tested against large sets of PAM measurements.

Overall, the PAB scenario provides a simple setting in which communication constraints, broadcasting maps, and multipartite nonclassical correlations can be studied together. The present work focused mainly on the simplest nontrivial case of two receivers, binary measurements, binary outcomes, and qubit messages. The reason is the computational complexity of the problem. However, a natural next step is to extend the analysis to more receivers, higher-dimensional communicated systems, larger input alphabets, and measurements with more outcomes. Another interesting scenario would be to allow inputs to the broadcasting channel, a situation in which the classical-classical $\mathcal{CC}$ and quantum-classical $\mathcal{QC}$ sets no longer coincide. Such extensions may reveal new multipartite structures and stronger forms of nonclassicality that do not appear in ordinary prepare-and-measure or Bell scenarios alone.

From a more applied perspective, PAB scenarios also provide a natural framework for semi-device-independent information-processing tasks. In a companion work~\cite{sarubi2026randomn}, we studied distributed $2\to1$ QRACs, characterized the trade-off between the information gains of the two receivers, and investigated semi-device-independent randomness certification in the same framework. This provides a broadcast counterpart to analogous information trade-offs studied in sequential prepare-transform-measure scenarios. Other possible directions include multipartite dimension witnessing, certification of broadcasting channels, and semi-device-independent cryptographic protocols.

\acknowledgements
We thank Armin Tavakoli and Carles Roch i Carceller for discussions and feedback on the manuscript. We acknowledge financial support from the EU Horizon Europe (QSNP, grant no. 101114043), the Danish National Research Foundation grant bigQ (DNRF 142), the Simons Foundation (Grant No. 1023171, R.C.), the Brazilian National Council for Scientific and Technological Development (CNPq, Grants No. 403181/2024-0 and 301687/2025-0), the National Institute of Science and Technology for Applied Quantum Computing through CNPq process No. 408884/2024-0, the Financiadora de Estudos e Projetos (Grant No. 1699/24 IIF-FINEP), a guest professorship from the Otto M\o{}nsted Foundation, and the Coordenação de Aperfeiçoamento de Pessoal de Nível Superior -- Brasil (CAPES) -- Finance Code 001. We also thank the High-Performance Computing Center (NPAD) at UFRN for providing computational resources. This work has been partially funded by the project “Design e arquiteturas de Protocolos CV-QKD” supported by the EMBRAPII CIMATEC Competence Center in Quantum Technologies-Quantum Industrial Innovation (QuIIN), with financial resources from the PPI IoT/Manufatura 4.0 of the MCTI Grant No. 053/2023, signed with EMBRAPII.
\newpage

\bibliography{2-references}

\appendix

\section{Convergence of the $\mathcal{QQ}$ SDP hierarchy}
\label{app:convergence}
Our proof follows the strategy of Sec.~4.1 of Ref.~\cite{navascues2008convergent}, with the main difference that our moment matrix also contains preparation monomials and the frame indices associated with the $d$-dimensional input space. The idea is still the same. We first show that the finite-level moment matrices stay uniformly bounded, so that a convergent subsequence defines a consistent infinite moment matrix. We then use this limiting matrix to reconstruct the measurements and the preparation states. The only extra step is that the preparation amplitudes first appear as commuting operators, which are then reduced to ordinary complex amplitudes through their joint spectral decomposition. In this way, every behavior that survives all levels of the hierarchy is written as a convex mixture of commuting-operator $\mathcal{QQ}$ behaviors, giving the closed convex hull of the commuting $\mathcal{QQ}$ set in the limit, as we shall see next.

\begin{definition}[Commuting-operator $\mathcal{QQ}$ set]
A behavior $p$ belongs to $\mathcal{QQ}^{\mathrm{com}}(d)$ if there exist a family of normalized states $\{|\psi_x\rangle\}\subset\mathbb C^d$, an isometry $U:\mathbb C^d\to\mathcal K$, with $U^\dagger U=\iden_d$, and sets of measurement operators $\{P^B_{b|y}\}_b$ and $\{P^C_{c|z}\}_c$ acting on $\mathcal K$, such that
\begin{equation}
    p(b,c|x,y,z)
    =
    \langle\psi_x|
    U^\dagger P^B_{b|y}P^C_{c|z}U
    |\psi_x\rangle ,
    \label{eq:qqcom_probability}
\end{equation}
with the measurement operators satisfying
\begin{enumerate}
    \item \textit{Hermiticity:}
    \begin{equation}
        (P^R_{a|s})^\dagger=P^R_{a|s}, \qquad R=B,C.
        \label{eq:qqcom_hermiticity}
    \end{equation}

    \item \textit{Orthogonality:}
    \begin{equation}
        P^R_{a|s}P^R_{a'|s}=\delta_{aa'}P^R_{a|s}, \qquad R=B,C.
        \label{eq:qqcom_projectivity}
    \end{equation}

    \item \textit{Completeness:}
    \begin{equation}
        \sum_a P^R_{a|s}=\iden, \qquad R=B,C.
        \label{eq:qqcom_completeness}
    \end{equation}

    \item \textit{Commutativity:}
    \begin{equation}
        [P^B_{b|y},P^C_{c|z}]=0.
        \label{eq:qqcom_commutation}
    \end{equation}
\end{enumerate}
Here $(a,s)$ denotes the output-input pair of the corresponding party.
\end{definition}

We denote by $\beta_{\mathcal{QQ}}^{\mathrm{com}}(d)$ the supremum of the witness value over behaviors in $\mathcal{QQ}^{\mathrm{com}}(d)$.

\begin{theorem}[Convergence of the hierarchy]
The SDP behavior sets form a nested family of outer approximations,
\begin{equation}
    \mathcal{QQ}^{\mathrm{com}}(d)
    \subseteq
    \mathcal Q_{q',w'}(d)
    \subseteq
    \mathcal Q_{q,w}(d),
    \quad q'\geq q,\quad w'\geq w,
    \label{eq:nesting}
\end{equation}
and their limiting intersection satisfies
\begin{equation}
    \bigcap_{q,w}\mathcal Q_{q,w}(d)
    =
    \overline{\operatorname{conv}}
    \bigl(\mathcal{QQ}^{\mathrm{com}}(d)\bigr).
    \label{eq:qq_set_convergence}
\end{equation}
Consequently, for every linear witness,
\begin{equation}
    \lim_{q,w\to\infty}\beta_{\mathcal{QQ}}^{(q,w)}
    =
    \beta_{\mathcal{QQ}}^{\mathrm{com}}(d).
    \label{eq:beta}
\end{equation}
\end{theorem}

\begin{proof}
The nesting in Eq.~\eqref{eq:nesting} is immediate, since higher levels only add moments and constraints while containing all lower-level ones. Next, we prove \eqref{eq:qq_set_convergence} and \eqref{eq:beta}.

\paragraph{Construction of the limiting representation.}
We begin by showing that all entries of the certificates are uniformly bounded. For convenience, denote a diagonal entry of $\mathcal M_{q,w}$ by $D(\mu,u):=L_{q,w}(\mu^*\mu\,\eta^{ii}_{u,u})$. Since $\mathcal M_{q,w}\succeq0$, we have $D(\mu,u)\geq0$. Suppose first that $\mu=\alpha_{x,r}\mu'$. The localizing constraint~\eqref{eq:qq_prep_localizing} gives $D(\mu',u)=\sum_s D(\alpha_{x,s}\mu',u)$. Thus $D(\mu',u)$ is a sum of nonnegative diagonal entries, one of which is $D(\mu,u)$, corresponding to $s=r$. Therefore $D(\mu,u)\leq D(\mu',u)$. In other words, deleting one preparation factor from $\mu$ can only increase the corresponding diagonal entry. The same argument applies to the measurement word. If $u=P^R_{a|s}u'$, projectivity and completeness imply $D(\mu,u')=\sum_a D(\mu,P^R_{a|s}u')$, so again $D(\mu,u)\leq D(\mu,u')$. Hence, by successively deleting all preparation factors from $\mu$ and all measurement factors from $u$, we obtain $0\leq D(\mu,u)\leq D(1,\iden)=1$, where the last equality follows from Eq.~\eqref{eq:qq_frame_normalization}. Thus every diagonal entry of $\mathcal M_{q,w}$ lies in $[0,1]$. Finally, for any two indices $A=(\mu,i,u)$ and $B=(\kappa,j,v)$, the corresponding $2\times2$ principal submatrix is positive semidefinite. Its determinant is therefore nonnegative, which gives
\begin{align}
|\mathcal M_{q,w}[A,B]|^2
&\leq \mathcal M_{q,w}[A,A]\,\mathcal M_{q,w}[B,B]
\leq 1,
\label{eq:qq_uniform_bound}
\end{align}
so $|\mathcal M_{q,w}[A,B]|\leq1$. Therefore every entry of $\mathcal M_{q,w}$ present at level $(q,w)$ is uniformly bounded by $1$.

We now construct the infinite moment matrix, following the argument in~\cite{navascues2008convergent}. Let
\begin{equation}
    p\in\bigcap_{n=1}^{\infty}\mathcal Q_{n,n}(d)
\end{equation}
be arbitrary, where it is sufficient to consider the diagonal sequence $q=w=n$, since the hierarchy is nested and, for every $(q,w)$, there exists $n\geq\max\{q,w\}$ such that
$\mathcal Q_{n,n}(d)\subseteq\mathcal Q_{q,w}(d)$. Since $p$ belongs to every diagonal level, for each $n$ there exists a feasible moment matrix $\mathcal M_{n,n}$ at level $(n,n)$ whose probability moments reproduce the same behavior $p$. As in the NPA proof, complete each $\mathcal M_{n,n}$ with zeros outside the indices already present at level $(n,n)$, obtaining an infinite matrix $\widehat{\mathcal M}_{n,n}$. By Eq.~\eqref{eq:qq_uniform_bound}, all entries of these matrices are uniformly bounded by $1$. The same compactness argument of Ref.~\cite{navascues2008convergent} therefore gives an increasing sequence of levels $\{n_k\}_{k\geq1}$, with $n_1<n_2<\cdots$, along which the completed matrices converge entrywise to an infinite matrix $\mathcal M_\infty$. Thus,
\begin{equation}
    \mathcal M_\infty[A,B]
    =
    \lim_{k\to\infty}
    \widehat{\mathcal M}_{n_k,n_k}[A,B],
    \label{eq:qq_infinite_matrix}
\end{equation}
and because the convergence holds for every fixed pair of indices, all relations involving finitely many moments pass directly to $\mathcal M_\infty$. Hence Eq.~\eqref{eq:qq_infinite_matrix} implies that $\mathcal M_\infty$ satisfies Eqs.~\eqref{eq:qq_scaled_orthonormality}, \eqref{eq:qq_prep_localizing}, \eqref{eq:qq_scalar_constraints}, and \eqref{eq:qq_word_constraints}. Moreover, let $\widehat{\mathcal M}_{n,n}^{(N)}$ denote the submatrix of $\widehat{\mathcal M}_{n,n}$ corresponding to the indices appearing up to level $(N,N)$, and define $\mathcal M_\infty^{(N)}$ analogously. Since $\widehat{\mathcal M}_{n,n}^{(N)}\succeq0$ for all $n$ and $N$, Eq.~\eqref{eq:qq_infinite_matrix} implies $\mathcal M_\infty^{(N)}\succeq0$. Lastly, since every $\mathcal M_{n,n}$ reproduces the same behavior $p$, the corresponding probability moments of $\mathcal M_\infty$ also reproduce $p$.

Therefore, exactly as in the NPA construction, there exists a family of vectors $\{|\mu,i,u\rangle\}$ whose scalar products reproduce the entries of $\mathcal M_\infty$,
\begin{equation}
    \mathcal M_\infty[(\mu,i,u),(\kappa,j,v)]
    =
    \langle \mu,i,u|\kappa,j,v\rangle .
    \label{eq:qq_infinite_gram}
\end{equation}
We now take as our Hilbert space $\mathcal H$ the closed linear span of these vectors.

\paragraph{Reconstruction of the quantum realization.}
We first reconstruct the measurement operators. We shall use projectors as the words here, although in the main text the hierarchy is formulated in terms of operator words. The main difference is that the reduction rules entering Eq.~\eqref{eq:qq_word_constraints} must now be replaced by the corresponding projector relations (hermiticity, orthogonality, completeness and commutativity). Accordingly, Eq.~\eqref{eq:qq_word_constraints} is understood with words reduced according to these projector relations and the corresponding commutation rule.

For each single-letter measurement word $P^R_{a|s}$, with $R\in\{B,C\}$, corresponding to outcome $a$ of setting $s$, define
\begin{equation}
    \widehat P^R_{a|s}
    =
    \operatorname{proj}\!\left(
    \operatorname{span}
    \left\{
    |\mu,i,P^R_{a|s}u\rangle:
    \mu,i,u
    \right\}
    \right),
    \label{eq:qq_reconstructed_projector}
\end{equation}
where $u$ ranges over all finite measurement words appearing in the limiting moment matrix. Thus, $\widehat P^R_{a|s}$ is the orthogonal projector onto the subspace generated by all vectors whose measurement-word label begins with $P^R_{a|s}$. From Eqs.~\eqref{eq:qq_word_constraints} and \eqref{eq:qq_infinite_gram}, for two outcomes $a,a'$ of the same measurement,
\begin{equation}
    \langle \mu,i,P^R_{a|s}u|\kappa,j,P^R_{a'|s}v\rangle
    =
    \delta_{aa'}
    \langle \mu,i,u|\kappa,j,P^R_{a|s}v\rangle ,
    \label{eq:qq_projector_inner_product}
\end{equation}
because the corresponding word product reduces as
$u^\dagger P^R_{a|s}P^R_{a'|s}v
=\delta_{aa'}u^\dagger P^R_{a|s}v$. For $a\neq a'$, Eq.~\eqref{eq:qq_projector_inner_product} therefore shows that every vector in the range of $\widehat P^R_{a|s}$ is orthogonal to every vector in the range of $\widehat P^R_{a'|s}$. Consequently,
\begin{equation}
    \widehat P^R_{a|s}|\mu,i,P^R_{a'|s}u\rangle
    =
    \delta_{aa'}|\mu,i,P^R_{a'|s}u\rangle ,
    \label{eq:P_delta}
\end{equation}
since an orthogonal projector leaves vectors in its own range unchanged and annihilates vectors in an orthogonal range. It follows that
\begin{equation}
    \widehat P^R_{a|s}\widehat P^R_{a'|s}
    =
    \delta_{aa'}\widehat P^R_{a|s},
    \label{eq:qq_reconstructed_projectivity}
\end{equation}
from which we obtain Eq.~\eqref{eq:qqcom_projectivity}. Furthermore, $\widehat P^R_{a|s}$ is Hermitian by construction, since it is an orthogonal projector, from which Eq.~\eqref{eq:qqcom_hermiticity} is obtained.

Let us now determine the action of $\widehat P^R_{a|s}$ on an arbitrary vector $|\mu,i,u\rangle$. From Eq.~\eqref{eq:qq_projector_inner_product},
\begin{equation}
    \langle \kappa,j,P^R_{a|s}v|\mu,i,u\rangle
    =
    \langle \kappa,j,P^R_{a|s}v|\mu,i,P^R_{a|s}u\rangle,
\end{equation}
so $|\mu,i,u\rangle-|\mu,i,P^R_{a|s}u\rangle$ is orthogonal to the range of $\widehat P^R_{a|s}$. Together with Eq.~\eqref{eq:P_delta}, this gives
\begin{equation}
    \widehat P^R_{a|s}|\mu,i,u\rangle
    =
    |\mu,i,P^R_{a|s}u\rangle.
    \label{eq:qq_reconstructed_measurement_action}
\end{equation}
Using the completeness relation in the word constraints, for any spanning vector $|\kappa,j,v\rangle$,
\begin{equation}
    \left\langle \kappa,j,v\middle|
    \sum_a\widehat P^R_{a|s}
    \middle|\mu,i,u\right\rangle
    =
    \sum_a
    \langle\kappa,j,v|\mu,i,P^R_{a|s}u\rangle
    =
    \langle\kappa,j,v|\mu,i,u\rangle .
\end{equation}
Since this holds for every spanning vector, $\sum_a\widehat P^R_{a|s}=\iden$, proving Eq.~\eqref{eq:qqcom_completeness}.

To recover the commutation relation, applying Eq.~\eqref{eq:qq_reconstructed_measurement_action} twice and using Bob--Charlie commutation in the word constraints gives, for every spanning vector $|\kappa,j,v\rangle$,
\begin{equation}
    \langle\kappa,j,v|
    \widehat P^B_{b|y}\widehat P^C_{c|z}
    |\mu,i,u\rangle
    =
    \langle\kappa,j,v|
    \widehat P^C_{c|z}\widehat P^B_{b|y}
    |\mu,i,u\rangle,
\end{equation}
and hence $[\widehat P^B_{b|y},\widehat P^C_{c|z}]=0$, proving Eq.~\eqref{eq:qqcom_commutation}.

We next reconstruct the preparations. For each single-letter preparation word $\alpha_{x,r}$, define $A_{x,r}$ on the span of the vectors by
\begin{equation}
    A_{x,r}|\mu,i,u\rangle
    =
    |\alpha_{x,r}\mu,i,u\rangle .
    \label{eq:qq_reconstructed_preparation_action}
\end{equation}
For an arbitrary vector in this span, write $|\xi\rangle=\sum_{\ell=1}^N c_\ell|\mu_\ell,i_\ell,u_\ell\rangle$. Using Eq.~\eqref{eq:qq_infinite_gram} and the limiting localizing relation~\eqref{eq:qq_prep_localizing},
\begin{align}
    \sum_r\|A_{x,r}\xi\|^2
    &=
    \sum_{\ell,m}\overline{c_\ell}c_m
    \sum_r
    \langle \alpha_{x,r}\mu_\ell,i_\ell,u_\ell
    |
    \alpha_{x,r}\mu_m,i_m,u_m\rangle
    \nonumber\\
    &=
    \sum_{\ell,m}\overline{c_\ell}c_m
    \langle \mu_\ell,i_\ell,u_\ell
    |
    \mu_m,i_m,u_m\rangle
    =
    \|\xi\|^2 .
    \label{eq:qq_reconstructed_prep_normalization}
\end{align}
In particular, if $|\xi\rangle=0$, then $\sum_r\|A_{x,r}\xi\|^2=0$, and hence $A_{x,r}|\xi\rangle=0$ for every $r$. Thus, the definition of $A_{x,r}$ is independent of the chosen linear representation of $|\xi\rangle$. Notice also that $\|A_{x,r}\xi\|\leq\|\xi\|$, so each $A_{x,r}$ extends to a bounded operator on $\mathcal H$, and Eq.~\eqref{eq:qq_reconstructed_prep_normalization} gives
\begin{equation}
    \sum_rA_{x,r}^\dagger A_{x,r}=\iden,
    \label{eq:qq_reconstructed_prep_identity}
\end{equation}
since $\langle\xi|\sum_rA_{x,r}^\dagger A_{x,r}|\xi\rangle=\|\xi\|^2$ for every $|\xi\rangle$.

To identify the algebra of these operators, first note that multiplication of a preparation label by $\alpha_{x,r}^*$ represents $A_{x,r}^\dagger$. Indeed,
\begin{equation}
    \langle \kappa,j,v|A_{x,r}|\mu,i,u\rangle
    =
    \langle \alpha_{x,r}^*\kappa,j,v|\mu,i,u\rangle ,
\end{equation}
since both sides correspond through Eq.~\eqref{eq:qq_infinite_gram} to the same preparation monomial $\kappa^*\alpha_{x,r}\mu$. The scalar constraints in Eq.~\eqref{eq:qq_scalar_constraints} then imply that multiplication by the variables $\alpha_{x,r}$ and $\alpha_{x,r}^*$ can be freely reordered. Hence
\begin{equation}
    [A_{x,r},A_{x',r'}]
    =
    [A_{x,r},A_{x',r'}^\dagger]
    =
    0 ,
\end{equation}
so the $A_{x,r}$ form a commuting family of normal operators.

The preparation and measurement actions also commute. Indeed, Eqs.~\eqref{eq:qq_reconstructed_measurement_action} and \eqref{eq:qq_reconstructed_preparation_action} give
\begin{equation}
    A_{x,r}\widehat P^R_{a|s}|\mu,i,u\rangle
    =
    \widehat P^R_{a|s}A_{x,r}|\mu,i,u\rangle
    =
    |\alpha_{x,r}\mu,i,P^R_{a|s}u\rangle ,
\end{equation}
and therefore $[A_{x,r},\widehat P^R_{a|s}]=0$.

We may now apply the joint spectral theorem to the commuting normal family $\{A_{x,r}\}$. Since the reconstructed measurement operators commute with every $A_{x,r}$ and $A_{x,r}^\dagger$, and therefore with the abelian von Neumann algebra generated by them, there exist a measure space $(\Omega,\gamma)$ and Hilbert spaces $\mathcal H_\omega$ such that
\begin{equation}
    \mathcal H\simeq\int_\Omega^\oplus\mathcal H_\omega\,d\gamma(\omega),
    \qquad
    A_{x,r}\simeq\int_\Omega^\oplus
    \alpha_{x,r}(\omega)\iden_{\mathcal H_\omega}\,d\gamma(\omega),
\end{equation}
with the reconstructed measurement operators acting fiberwise on the same decomposition. Eq.~\eqref{eq:qq_reconstructed_prep_identity} then holds fiberwise and gives
$\sum_{r=1}^d|\alpha_{x,r}(\omega)|^2=1$ for almost every $\omega$, so
\begin{equation}
    |\psi_x(\omega)\rangle
    =
    \sum_{r=1}^d\alpha_{x,r}(\omega)|r\rangle
    \in\mathbb C^d
    \label{eq:qq_reconstructed_state}
\end{equation}
is a normalized state for almost every $\omega$.

It remains to reconstruct the isometry. Let $|e_i\rangle:=|1,i,\iden\rangle$, and denote its component in $\mathcal H_\omega$ by $|e_i(\omega)\rangle$. Eq.~\eqref{eq:qq_scaled_orthonormality} gives, for every preparation polynomial $f$,
\begin{equation}
    \langle e_i|f(A,A^\dagger)|e_j\rangle
    =
    \delta_{ij}
    \langle e_1|f(A,A^\dagger)|e_1\rangle,
\end{equation}
which, in the joint spectral representation, means
\begin{equation}
    \int_\Omega
    f(\alpha(\omega),\overline{\alpha(\omega)})
    \left(
    \langle e_i(\omega)|e_j(\omega)\rangle
    -
    \delta_{ij}\|e_1(\omega)\|^2
    \right)
    d\gamma(\omega)
    =
    0.
\end{equation}
Since polynomials in the spectral variables and their conjugates are dense in the continuous functions on the compact joint spectrum, the term in parentheses must vanish almost everywhere. Hence
\begin{equation}
    \langle e_i(\omega)|e_j(\omega)\rangle
    =
    \delta_{ij}\|e_1(\omega)\|^2
    \label{eq:qq_fiber_frame}
\end{equation}
for almost every $\omega$, and whenever $\|e_1(\omega)\|>0$ the vectors
$|\widetilde e_i(\omega)\rangle:=|e_i(\omega)\rangle/\|e_1(\omega)\|$
are orthonormal. Defining
\begin{equation}
    U_\omega|i\rangle
    =
    |\widetilde e_i(\omega)\rangle
\end{equation}
therefore gives an isometry $U_\omega:\mathbb C^d\to\mathcal H_\omega$, with
$U_\omega^\dagger U_\omega=\iden_d$.

\paragraph{Recovery of the behavior and convergence.}
The measurement identities proved above also hold fiberwise almost everywhere. Thus, for almost every $\omega$ with $\|e_1(\omega)\|>0$, the normalized states $|\psi_x(\omega)\rangle$, the isometry $U_\omega$, and the operators $\widehat P^R_{a|s}(\omega)$ satisfy all the defining properties of $\mathcal{QQ}^{\mathrm{com}}(d)$. In particular,
\begin{equation}
    p_\omega(b,c|x,y,z)
    =
    \langle\psi_x(\omega)|
    U_\omega^\dagger
    \widehat P^B_{b|y}(\omega)
    \widehat P^C_{c|z}(\omega)
    U_\omega
    |\psi_x(\omega)\rangle
    \label{eq:qq_component_behavior}
\end{equation}
defines a commuting-$\mathcal{QQ}$ behavior.

We can now relate these component behaviors to the original behavior $p$. Since
$|e_i(\omega)\rangle=\|e_1(\omega)\|\,|\widetilde e_i(\omega)\rangle$, the probability moments reproduced by $\mathcal M_\infty$ satisfy
\begin{equation}
    p(b,c|x,y,z)
    =
    \int_\Omega
    \|e_1(\omega)\|^2
    p_\omega(b,c|x,y,z)\,
    d\gamma(\omega).
\end{equation}
Introducing
$d\nu(\omega):=\|e_1(\omega)\|^2d\gamma(\omega)$, Eq.~\eqref{eq:qq_frame_normalization} gives
$
    \int_\Omega d\nu(\omega)
    =
    \|e_1\|^2
    =
    1
$, so $\nu$ is a probability measure and the previous identity becomes
\begin{equation}
    p(b,c|x,y,z)
    =
    \int_\Omega
    p_\omega(b,c|x,y,z)\,d\nu(\omega).
    \label{eq:qq_behavior_decomposition}
\end{equation}
Since $p_\omega\in\mathcal{QQ}^{\mathrm{com}}(d)$ for almost every $\omega$, Eq.~\eqref{eq:qq_behavior_decomposition} shows that
\begin{equation}
    p
    \in
    \overline{\operatorname{conv}}
    \bigl(\mathcal{QQ}^{\mathrm{com}}(d)\bigr).
\end{equation}
As $p$ was arbitrary, this gives
\begin{equation}
    \bigcap_{q,w}\mathcal Q_{q,w}(d)
    \subseteq
    \overline{\operatorname{conv}}
    \bigl(\mathcal{QQ}^{\mathrm{com}}(d)\bigr).
\end{equation}
The converse inclusion follows because every commuting-$\mathcal{QQ}$ behavior is feasible at every finite level, so $\mathcal{QQ}^{\mathrm{com}}(d)\subseteq\mathcal Q_{q,w}(d)$, and since each $\mathcal Q_{q,w}(d)$ is convex, it also contains $\operatorname{conv} \bigl(\mathcal{QQ}^{\mathrm{com}}(d)\bigr)$. Therefore
\begin{equation}
    \bigcap_{q,w}\mathcal Q_{q,w}(d)
    =
    \overline{\operatorname{conv}}
    \bigl(\mathcal{QQ}^{\mathrm{com}}(d)\bigr),
\end{equation}
which proves Eq.~\eqref{eq:qq_set_convergence}. 

To conclude the proof, since the SDP behavior sets are nested and compact, the maximal value of a linear witness converges to its maximal value over their limiting intersection. By Eq.~\eqref{eq:qq_set_convergence}, this intersection is the closed convex hull of the commuting-$\mathcal{QQ}$ set, over which any linear witness has the same optimal value as over the original set. Hence
\begin{equation}
    \lim_{q,w\to\infty}\beta_{\mathcal{QQ}}^{(q,w)}
    =
    \beta_{\mathcal{QQ}}^{\mathrm{com}}(d),
\end{equation}
which proves Eq.~\eqref{eq:beta}, concluding the proof.
\end{proof}

\section{Optimal configurations}
\label{Optimal_configurations}
The SDP relaxation of Sec.~\ref{subsec:qqineqs} provides upper bounds on the witness values reported in Table~\ref{tab:qq_witness_values}. Here, we complement these results with explicit quantum strategies providing corresponding lower bounds. We follow the optimization approach introduced in Sec.~\ref{sec:quantum-broadcasting}, considering pure qubit preparations $\rho_x=\frac{1}{2}(\iden+\vec r_x\cdot\vec\sigma)$, a general qubit-to-two-qubit broadcasting isometry such that $\tau_x^{BC}=U\rho_xU^\dagger$, and dichotomic projective measurements $B_y=\vec b_y\cdot\vec\sigma$ and $C_z=\vec c_z\cdot\vec\sigma$ for Bob and Charlie, respectively. The corresponding two-body correlators are
\begin{equation}
\langle B_yC_z\rangle^{[x]}
=
\tr\left[\tau_x^{BC}(B_y\otimes C_z)\right].
\end{equation}

\subsection{Configuration for \texorpdfstring{$W_{\mathcal{CC}}^{(1)}$}{WCC(1)}}
\label{app:config_wcc1}

For the inequality $W_{\mathcal{CC}}^{(1)}$ in Eq.~\eqref{eq:CC}, we obtain
$(W_{\mathcal{CC}}^{(1)})^Q=12.944$. One strategy attaining this value is
\begin{equation}
\begin{aligned}
\vec r_0 &= (-0.478199,-0.107883,-0.871600),\\
\vec r_1 &= (+0.250024,+0.932540,+0.260494),\\
\vec r_2 &= (+0.254570,-0.726206,+0.638607),\\[2mm]
\vec b_0 &= (-0.082521,+0.038075,+0.995862),\\
\vec b_1 &= (-0.173429,-0.984571,+0.023272),\\
\vec c_0 &= (-0.548986,-0.011222,-0.835756),\\
\vec c_1 &= (-0.743867,+0.462533,+0.482416).
\end{aligned}
\end{equation}
The associated isometry, written in the computational basis, is
\begin{equation}
U_{\mathrm{opt}}^{(1)}
=
\begin{pmatrix}
 0.435343+0.426970i & -0.100139+0.326535i \\
 0.137877+0.360192i &  0.334318-0.501002i \\
 0.193901-0.170731i & -0.382132-0.546579i \\
-0.081218+0.637245i & -0.275280-0.004776i
\end{pmatrix}.
\end{equation}
Evaluating Eq.~\eqref{eq:CC} with these parameters gives the value
reported above.

\subsection{Configuration for \texorpdfstring{$W_{\mathcal{CC}}^{(2)}$}{WCC(2)}}
\label{app:config_wcc2}

For the inequality $W_{\mathcal{CC}}^{(2)}$ in Eq.~\eqref{eq:CC2}, the optimized
value is $(W_{\mathcal{CC}}^{(2)})^Q=10.828$. A corresponding numerical strategy is
\begin{equation}
\begin{aligned}
\vec r_0 &= (+0.346743,-0.937953,-0.003529),\\
\vec r_1 &= (-0.002956,+0.002670,-0.999992),\\
\vec r_2 &= (-0.243094,+0.661345,+0.709597),\\[2mm]
\vec b_0 &= (-0.617353,-0.151021,-0.772054),\\
\vec b_1 &= (+0.562743,+0.601009,-0.567546),\\
\vec c_0 &= (+0.628887,+0.492636,-0.601507),\\
\vec c_1 &= (+0.171344,+0.666810,+0.725263).
\end{aligned}
\end{equation}
The optimized isometry is
\begin{equation}
U_{\mathrm{opt}}^{(2)}
=
\begin{pmatrix}
 0.302656-0.638238i &  0.039246-0.217752i \\
 0.024188+0.014229i & -0.355177-0.570226i \\
-0.009192+0.025650i &  0.517734+0.428141i \\
 0.651198+0.274706i & -0.219738+0.009898i
\end{pmatrix}.
\end{equation}
Substitution into Eq.~\eqref{eq:CC2} gives the value quoted in
Table~\ref{tab:qq_witness_values}.

\subsection{Configuration for \texorpdfstring{$W_{\mathcal{CNS}}^{(3)}$}{WCNS(3)}}
\label{app:config_wcns3}

For the inequality $W_{\mathcal{CNS}}^{(3)}$ in
Eq.~\eqref{eq:WNS3}, we find
$(W_{\mathcal{CNS}}^{(3)})^Q=5.196$. One
configuration realizing this value is
\begin{equation}
\begin{aligned}
\vec r_0 &= (+0.866026,+0.000000,-0.500000),\\
\vec r_1 &= (+0.000000,+0.000000,+1.000000),\\
\vec r_2 &= (-0.866025,+0.000000,-0.500000),\\[2mm]
\vec b_0 &= (-0.206988,+0.940398,+0.269828),\\
\vec b_1 &= (-0.702801,-0.707550,+0.073780),\\
\vec c_0 &= (-0.439119,-0.242599,+0.865055),\\
\vec c_1 &= (-0.558163,+0.199290,-0.805442).
\end{aligned}
\end{equation}
The corresponding isometry is
\begin{equation}
U_{\mathrm{opt}}^{\mathcal{CNS},3}
=
\begin{pmatrix}
-0.406025-0.219587i &  0.013242-0.408367i \\
-0.086760+0.528581i & -0.522423-0.245225i \\
 0.226016+0.485636i &  0.436995-0.376957i \\
-0.331814+0.320896i &  0.122818+0.389685i
\end{pmatrix}.
\end{equation}
These parameters reproduce the value reported for
$W_{\mathcal{CNS}}^{(3)}$.

\end{document}